\documentclass[a4paper]{article}

\usepackage{amsmath,amsthm,amssymb,mathtools}
\usepackage{comment}
\usepackage[T1]{fontenc} 
\usepackage[mathcal]{euscript}
\usepackage{graphicx}
\usepackage[hidelinks]{hyperref}
\usepackage[top=2.7cm,bottom=2.7cm,left=2.4cm,right=2.4cm]{geometry}
\usepackage{tikz}
\usepackage[export]{adjustbox}
\usepackage{appendix}
\usepackage{setspace}
\usepackage{tcolorbox}
\usepackage{enumitem}
\numberwithin{equation}{section}
\theoremstyle{plain}

\theoremstyle{definition}
\newtheorem{remark}{Remark}[section]
\newtheorem{definition}{Definition}[section]
\newtheorem{proposition}{Proposition}[section]
\newtheorem{example}{Example}[section]

\makeatletter
\newcommand*{\rom}[1]{\expandafter\@slowromancap\romannumeral #1@}
\makeatother
\makeatletter
\g@addto@macro\bfseries{\boldmath}
\makeatother

\usetikzlibrary{shapes.geometric,decorations.markings,intersections,calc,positioning}

\tikzset{->-/.style={decoration={
  markings,
  mark=at position .6 with {\arrow{>}}},postaction={decorate}}}

\tikzset{-<-/.style={decoration={
  markings,
  mark=at position .6 with {\arrow{<}}},postaction={decorate}}}

\tikzset{%
    add/.style args={#1 and #2}{
        to path={%
 ($(\tikztostart)!-#1!(\tikztotarget)$)--($(\tikztotarget)!-#2!(\tikztostart)$)%
  \tikztonodes},add/.default={.2 and .2}}
}  

\tikzset{
    extended line/.style={shorten >=-#1,shorten <=-#1},
    extended line/.default=1cm]
}

\tikzset{line through/.style args={#1 parallel to line through #2 and #3 and
length #4}{insert path={%
let \p1=($(#3)-(#2)$),\n1={atan2(\y1,\x1)} in (#1) -- ++ (\n1:#4)}}}

\definecolor{col1}{rgb}{0.4, 0.69, 0.2}%cyan(process)
\definecolor{col2}{rgb}{0.96, 0.29, 0.54}%fulvous

\definecolor{green(ryb)}{rgb}{0.4, 0.69, 0.2}
\definecolor{frenchrose}{rgb}{0.96, 0.29, 0.54}
\definecolor{persianblue}{rgb}{0.11, 0.22, 0.73}
\definecolor{jade}{rgb}{0.0, 0.66, 0.42}
\definecolor{limegreen}{rgb}{0.2, 0.8, 0.2}

\title{Generalisation of Bureau-Guillot systems with \\ Painlev\'e transcendents in the coefficients }
\author{Marta Dell'Atti{\color{jade}$\,^{1,*}$} \\[-.5ex] \small{\url{m.dell-atti@uw.edu.pl}} \\[1ex] Galina Filipuk{\color{jade}$\,^1$} \\[-.5ex]  \small{\url{g.filipuk@uw.edu.pl}} \\[1ex]
{\color{jade}$\,^1$}{\small University of Warsaw, Institute of Mathematics, 
 Banacha 2, Warsaw 02-097, Poland}
\\[.2ex]
{\color{jade}$\,^*$}{\small Corresponding author}
}

\date{}

\begin{document}

\maketitle

\begin{abstract}
 We construct a generalisation of what we call Bureau-Guillot systems, i.e.\ systems of first order equations with coefficient functions being Painlev\'e transcendents. The same Painlev\'e equation is related to the system and it appears as   regularising  condition in the regularisation process. The systems considered are birationally equivalent to the Okamoto polynomial Hamiltonian systems with rational coefficients for Painlev\'e equations,  hence they possess the Painlev\'e property. This work extends the results of Bureau-Guillot in a two-fold way. On one side, we consider polynomial systems with degree larger than $2$ that are free of movable critical points. These systems contain not only transcendents $\text{P}_{\text{I}}$ and $\text{P}_{\text{II}}$ in the coefficients, but also transcendents~$\text{P}_{\text{III}}$, $\text{P}_{\text{IV}}$, $\text{P}_{\text{V}}$ and $\text{P}_{\text{VI}}$ (and/or their derivatives). On the other side, we explore examples of rational systems   with the Painlev\'e transcendents in the coefficients birationally equivalent to the Okamoto polynomial systems. 
 Lastly, we present a simpler version of the change of variables to obtain  the analogues of the  Bureau-Guillot systems. In this framework we discuss generalisations including the mixed case:   systems related to the equation $(\text{P}_{\text{J}})$, with ${\text{J}=\text{I}, \dots, \text{VI}}$, containing  coefficient functions that are solutions to $(\text{P}_{\text{K}})$ with~$\text{K}\neq \text{J}$.  In the latter, the equation $(\text{P}_{\text{K}})$ appears as a regularising condition during  the regularisation process.   Although we are primarily interested in systems possessing the Painlev\'e property, we also briefly discuss an  analogous construction for systems including coefficient functions solving quasi-Painlev\'e equation.

\end{abstract}



\vspace{2ex}

\section{Introduction}

     Painlev\'e equations are nonlinear differential equations of second order of the form 
    \begin{equation}\label{eq:rational_sec}
        y''=R\big(y,y';x\big)\,, \qquad y = y(x) \in \mathbb{C}, \quad x \in \mathbb{C}\,, 
    \end{equation}
    with $R$ being polynomial in $y'$, rational both in $y$ and $x$. They were studied by the Painlev\'e school~(\cite{Painleve1900,Gam})
    with the purpose of classifying the equations of the form~\eqref{eq:rational_sec} with what is now known as the Painlev\'e property: their solutions are free from movable critical points. The nonlinear nature of the equations induces the appearance of movable singularities, depending on the initial conditions, but for the Painlev\'e property these can be at most poles. The classification leads to the identification of  fifty standard forms of equations, among which  solutions of six equations are not expressible in terms of already known (classical) functions. These six equations are known as Painlev\'e equations, which we will label as $(\text{P}_{\text{J}})$ with $\text{J}= \text{I}, \dots, \text{VI}$.  We will introduce them later on. Their solutions, the Painlev\'e transcendents, are instead labelled in this paper as $\text{P}_{\text{J}}$ with $\text{J}= \text{I}, \dots, \text{VI}$. 

    {
    The Painlev\'e equations $(\text{P}_{\text{J}})$ with $\text{J}= \text{I}, \dots, \text{VI}$, admit a (non-unique) representation as non-autonomous Hamiltonian systems of two first order differential equations, as studied in the classical works by Okamoto~\cite{okamoto1,okamoto2,okamoto3,okamoto4}. 
    We will consider systems of the type \begin{equation}\label{eq:original}
    \begin{cases}
        y'= P(y,z;x), \\
        z'= Q(y,z;x),
    \end{cases} \qquad (y,z) \in \mathbb{C}^2\,, ~~x \in \mathbb{C}\,, 
\end{equation}
with $y=y(x)$, $z=z(x)$ and $P$ and $Q$ polynomial (or rational) functions, that will be regularised according to the algebro-geometric Okamoto-Sakai method (see~\cite{Okamoto1979,Sakai2001}) that we will briefly recall, to construct the so-called space of initial conditions. First, to include the points at infinity, we need to consider a compactification of the complex plane $\mathbb{C}^2$, this being throughout the paper $\mathbb{P}^2$.   
    The compactified space $\mathbb{P}^2$ is given by gluing together three affine charts identified by $(y,z)$, $(u_0,v_0)$ and $(U_0,V_0)$, related by homogeneous coordinates $[w_0:w_1:w_2]$, such that 
\vspace*{-1ex}
\begin{equation}\label{eq:CP2}
\includegraphics[width=.23\textwidth,valign=c]{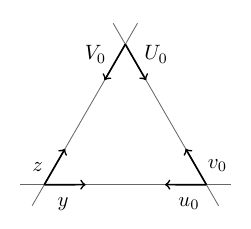} \qquad   
    \begin{aligned}
    &\mathbb{P}^2 = \mathbb{A}_{(y,z)} \cup  \mathbb{A}_{(u_0,v_0)} \cup \mathbb{A}_{(U_0,V_0)}\,, \\[1.4ex] 
    &[\,1:y:z\,] = [\,u_0:1:v_0\,] = [\,V_0:U_0:1\,]\,, \\[.7ex]
    &u_0 = \frac{1}{y}\,, \qquad V_0 = \frac{1}{z}\,, \qquad v_0 = \frac{z}{y} = \frac{1}{U_0} \,.
\end{aligned}
\end{equation} 

\vspace*{-1ex}

\noindent
In each affine chart, the system is represented as a pair of rational differential equations in the corresponding variables: 
\begin{equation}\label{eq:P1P1}
    \begin{cases}
        (u_0)'= P_1(u_0,v_0;x), \\[1ex]
        (v_0)'= Q_1(u_0,v_0;x),
    \end{cases} \qquad \begin{cases}
        (U_0)'= P_2(U_0,V_0;x), \\[1ex]
        (V_0)'= Q_2(U_0,V_0;x),
    \end{cases} 
\end{equation}
where the functions $P_k, Q_k$ for $k=1,2$ are in general rational functions. Analysing the systems in the three charts composing $\mathbb{P}^2$, we identify the so-called 
base points by looking for (finitely many) points of indeterminacy labelled as
$b_j$ ($j=1,\dots,N$), for which the right hand side of the rational differential equations becomes~$0/0$. These are points of coalescence of infinitely many integral curves, and to resolve the indeterminacy,  we introduce a blow-up transformation of the surface $\mathbb{P}^2$, such that the point is substituted by a complex projective line $\mathbb{P}^1$. In this way, each of the previously coalescing curves intersects the line $\mathbb{P}^1$ in one point only, as depicted heuristically here:
\begin{equation*}
    \includegraphics[width=0.45\linewidth]{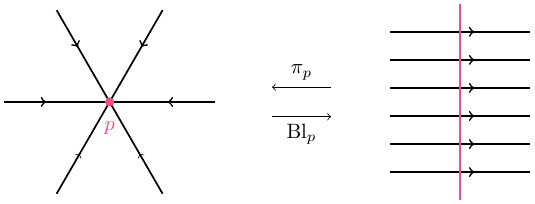}
\end{equation*}
The blow-up transformation analytically corresponds to a birational transformation between affine charts. In particular, the system $(x,y)$ blown up at the point $p$ with coordinates $p \colon (x,y)=(a,b)$ is such that the point is replaced by the projective line $\mathbb{P}^1$, completely described by two charts, say $(u,v)$ and $(U,V)$. The birational transformation is described via the changes of coordinates 
\begin{equation}\label{eq:blowup}
    \begin{cases}
        u=y-a\,, \\[1.2ex]
        v=\dfrac{z-b}{y-a} \,,
    \end{cases}  \hspace{6ex} \begin{cases}
        U=\dfrac{y-a}{z-b}\,, \\[1.2ex]
        V=z-b\,. 
    \end{cases} 
\end{equation}
For our scope, it is also worth considering the inverse of this transformation, i.e. 
\begin{equation}\label{eq:blowup_inverse}
    \begin{cases}
        y= u + a = UV + a\,, \\[1ex]
        z= uv +b = V + b\,.
    \end{cases}
\end{equation}
The line isomorphic to $\mathbb{P}^1$ emerging from this transformation is referred to as the exceptional curve~$E$ described by the sets
\begin{equation}\label{eq:exc_curve}
    E = \{u = 0\} \cup \{ V = 0 \}\,.  
\end{equation}
The system of differential equations in which the indeterminacy point is visible is then transformed in two charts, giving rise to two rational differential systems in the coordinates $(u,v)$ and $(U,V)$ respectively. These systems may have new points of indeterminacy, at which we apply again the blow-up transformation. The finite sequence of these transformations at points $b_i^{(j)}$ gives rise to the so-called cascade of points. 
We use the upper index $(j)$ to indicate the $j$-th cascade emerging from the point $b_i$ in the original system. 
In the $j$-th cascade the convention for the change of variables (omitting the upper index) $(u_i,v_i) \mapsto (u_{i+1},v_{i+1}) \cup (U_{i+1},V_{i+1})$ to perform the blow-up in the chart $(u_i,v_i)$ at the point with oordinates $(u_{i},v_i)=(a_i,b_i)$ is  
\begin{equation}
    \begin{cases}
        u_{i+1} = u_{i}-a_i, \\[1ex]
        v_{i+1} = \dfrac{v_i - b_i}{u_i - a_i},
    \end{cases} \qquad \begin{cases}
        U_{i+1} = \dfrac{u_{i} - a_i}{v_i - b_i}, \\[1.5ex]
        V_{i+1} = v_i + b_i. 
    \end{cases}
\end{equation}
Every time a blow-up is performed, an additional exceptional curve emerges, represented as the set $E_{i+1}=\{u_{i+1}=0\} \cup \{V_{i+1}=0\}$, adapting the formulation in~\eqref{eq:exc_curve}. This procedure is iterated until the systems in both final charts $(u_{\text{fin}},v_{\text{fin}})$ and $(U_{\text{fin}},V_{\text{fin}})$ are not affected by further indeterminacies. At the end of the procedure the system is said regularised if it is regular on the last exceptional curve in one of the final charts in every cascade. Hence, the regularisation process gives rise to a chain of finitely many consecutive birational transformations of variables, e.g.\ 
\begin{equation}\label{eq:chain}
    (y,z) ~\mapsto~ (u_1,v_1) ~\mapsto~ \dots~ \mapsto ~(U_k,V_k)~ \mapsto ~ \dots~ \mapsto ~(u_{\text{fin}},v_{\text{fin}}) \,. 
\end{equation}
The regularising conditions arising in the regularisation process, i.e.\ (differential) constraints on the coefficient functions appearing in the systems which make the system in one of the final charts regular   on the last exceptional divisors, will be of crucial importance throughout the paper\footnote{
Finding the regularisation conditions 
is equivalent to fixing the resonance conditions on the coefficients in the series expansion of solutions in the (quasi)-Painlev\'e test (see~\cite{GF_AS1,GF_AS2,TKGF,MDTK1,MDTK2}).}. The space of initial conditions is then given by removing the specific configuration of irreducible inaccessible divisors  from~$\mathbb{P}^2$ (or vertical leaves in the Okamoto's language~\cite{Okamoto1979}) forming the so-called surface types represented as the extended Dynkin diagrams $ADE$~\cite{Takano97,Takano99}.

In this paper, we will consider systems of the type~\eqref{eq:original}, possessing meromorphic coefficients in $x$ and endowed with the Painlev\'e property. We will start by analysing the cases originally studied by Bureau~\cite{Bureau} and revised by Guillot~\cite{Guillot} related to the first and the second Painlev\'e equations, where the coefficient functions are in turn solutions to ($\text{P}_{\text{I}}$) and ($\text{P}_{\text{II}}$) respectively, in the framework of quadratic systems free from critical points. In~\cite{Guillot} some of the birational relations connecting the systems pairwise are identified, and in~\cite{MDGF} the birational equivalence of the systems is contextualised within the Okamoto-Sakai approach. Then, we will proceed with an extension of the above mentioned systems, and finally determine a generalisation of the procedure to include the other Painlev\'e equations for polynomial  systems of degree greater than two. 

Here, we formalise the definition of what we mean by the analogue of the  Bureau-Guillot system. 
\begin{definition}\label{def:BG}
    A system as in~\eqref{eq:original} is an analogue of the Bureau-Guillot system if it has the following properties: 
    \begin{enumerate}[label=($\roman*$)]
        \item it is birationally equivalent to the Hamiltonian system with coefficients rational in $x$ associated with the
        Painlev\'e equation $(\text{P}_{\text{J}})$;
        \item it has coefficient functions both rational and  meromorphic in $x$ which are functions of the Painlev\'e transcendent $\text{P}_{\text{J}}$ and its derivatives;
        \item we obtain the Painlev\'e equation $(\text{P}_{\text{J}})$ as a regularising 
        condition in the regularisation process of the system. 
    \end{enumerate}
\end{definition}
\noindent 
Because of $(i)$ the system has the Painlev\'e property automatically. 

The Bureau-Guillot systems and their generalizations are important as they extend beyond equations with rational coefficients and movable singularities. While in this paper we consider specific problems of constructing the Bureau-Gulliot systems in the sense of the definition above, the general classification of differential equations with polynomial/rational right-hand sides that have the Painlev\'e property is far more complicated. The case of rational coefficients alone is challenging, and assuming meromorphic coefficients quickly  increases the complexity. 

    }

{
In the following, in Section~\ref{sec:BGgen} we will revisit the well-known quadratic Bureau-Guillot systems according to Definition~\ref{def:BG} associated with the Painlev\'e equations $(\text{P}_{\text{I}})$ and $(\text{P}_{\text{II}})$. We identify the mechanism to produce such systems and propose a novel generalisation of these systems. In Section~\ref{sec:analogueBG}, we construct the analogue of Bureau-Guillot systems for the polynomial systems associated with Painlev\'e equations~$(\text{P}_{\text{III}})$ (of types $D_6$, $D_7$ and   $D_8$), $(\text{P}_{\text{IV}})$, $(\text{P}_{\text{V}})$ and $(\text{P}_{\text{VI}})$. Lastly, in Section~\ref{sec:gen_mixedcases}, we provide an alternative construction to include a generalised class of systems. In this framework, we present examples of mixed cases: we construct systems related to $(\text{P}_{\text{J}})$ for $\text{J} = \text{I}, \dots, \text{VI}$, whose coefficient functions are solutions of~$(\text{P}_{\text{K}})$, with $\text{K} \neq \text{J}$, this appearing in the regularising conditions. The latter construction can be extended to the case of quasi-Painlev\'e equations as well. 
}

\section{\texorpdfstring{Bureau-Guillot systems and their generalizations}{BGgen}}\label{sec:BGgen}
{ 
In this Section, we deal with the  systems studied by Bureau and Guillot with the nomenclature introduced in~\cite{Bureau} associated with $(\text{P}_{\text{I}})$ and $(\text{P}_{\text{II}})$. We identify the type of birational transformations connecting the systems pairwise, and via a generalisation of this transformation we determine an extended (or parametrised) version of the above mentioned Bureau-Guillot systems. Moreover, we show that this procedure can be iterated. 

}
\subsection{\texorpdfstring{Systems V and (parametrised) IX.B(2) for $(\text{P}_{\text{I}})$}{P1}}
The system V~\cite{Bureau,Guillot} in coordinates $(y_5,z_5)$, with $y_5=y_5(x)$ and $z_5=z_5(x)$, is 
\begin{equation}\label{eq:system_5}
\begin{cases}
    y_5' = z_5,\\[1ex]
z_5' = 6\,y_5^2 + f(x),
\end{cases} \qquad f''(x) = 0\,.
\end{equation}
As shown in~\cite{MDGF}, the condition on the arbitrary function $f''=0$ can be recovered from the regularisation procedure: it coincides with the regularising condition for which the system V satisfies the Painlev\'e property. In particular, fixing the integration constants such that $f(x)=x$, the variable $y_5$ satisfies the equation $(\text{P}_{\text{I}})$
\begin{equation}\label{eq:P1}
    (\text{P}_{\text{I}})\colon ~~ y_5''= 6\,y_5^2 + x \,. 
\end{equation}

The system IX.B(2) in coordinates $(y_{92},z_{92})$,  with $y_{92}=y_{92}(x)$ and $z_{92}=z_{92}(x)$ is 
\begin{equation}\label{eq:system_92}
\begin{cases}
    y'_{92} =  z_{92} - y_{92}^2 +12 q(x)\,, \\[1ex]
    z_{92}' = y_{92}\,z_{92} ,
\end{cases}
    \qquad (q''-6q^2)''=0\,. 
\end{equation}
The  additional condition on $q(x)$ can be interpreted again as the regularising  condition for the system IX.B(2) to have the Painlev\'e property, and it appears at the end of the regularisation process for the system. 

As already noticed in~\cite{Guillot}, and contextualised within the Okamoto-Sakai approach in~\cite{MDGF}, the systems~V and IX.B(2) are both associated with the surface type $E_8^{(1)}$, and related by the following birational change of coordinates $(y_5,z_5) \mapsto (y_{92},z_{92})$: 
\begin{equation}\label{eq:birational_syst92_syst5}
z_{92}=6(y_5-q)\,, \qquad  y_{92}=\frac{z_5-q'}{y_5-q}\,. 
\end{equation}
If we look at the definition of the blow-up transformation~\eqref{eq:blowup}, we can interpret the birational transformation~\eqref{eq:birational_syst92_syst5} as   a sort of parametrised (or re-scaled) blow-up transformation from the system~\eqref{eq:system_5} in coordinates~$(y_5,z_5)$ at the point 
\begin{equation}\label{eq:point_syst_5}
p_{5}\colon (y_5,z_5)= \big(q,q'\big)\,.     
\end{equation}
What distinguishes the transformation of the type~\eqref{eq:birational_syst92_syst5} from a standard blow-up transformation is the presence of a re-scaling term in the expression of $z_{92}$. Moreover, $q$ is itself solution to $(\text{P}_{\text{I}})$ and the point~$p_5$ in~\eqref{eq:point_syst_5} is a point on the orbit of the system V in~\eqref{eq:system_5}. The blow-up is performed as in~\eqref{eq:blowup_inverse} by considering the coordinates of two charts, that we call $(u,v)$, and $(U,V)$
\begin{equation}\label{eq:blowup_syst5}
\begin{cases}
    y_5 = \dfrac{1}{6}\,u+q   = \dfrac{1}{6}\,UV + q \,,  \\[2ex]
    z_5 =  \dfrac{1}{6}\,uv+q'  = \dfrac{1}{6}\,V + q' \,.
\end{cases}
\end{equation}
Lastly, from the new system in the chart $(U,V)$ we perform a further final change of variables to obtain the system IX.B(2), i.e.
\begin{equation}\label{eq:inversion}
(U\,,V) \mapsto \big( y_{92} = U^{-1}\,, z_{92} = UV \big)   \,. 
\end{equation}
In coordinates $(y_{92},z_{92})$ the system has the form reported in~\eqref{eq:system_92}, and we find that $q$ is a $\text{P}_{\text{I}}$ transcendent as a regularising condition for the system. 

We generalise the blow-up-like transformation~\eqref{eq:blowup_syst5} with focus on the chart $(U,V)$, introducing two arbitrary constants $a_1, a_2 \in \mathbb{C}$ obtaining a modified version of the system IX.B(2), as in the following.

\begin{proposition}
The modified system IX.B(2) depending on two parameters $a_1,a_2 \in \mathbb{C}$ is 
\begin{equation}\label{eq:system_92_mod}
    \begin{cases}
    (y_{92\text{m}})' = -\dfrac{a_2}{a_1}(y_{92\text{m}})^2 + 6\,\dfrac{a_1}{a_2}\big(a_1\,z_{92\text{m}}+2q\big)   \,,  \\[2ex]
    (z_{92\text{m}})' = \dfrac{a_2}{a_1}\,y_{92\text{m}} z_{92\text{m}}   \,,
\end{cases}
\end{equation}
with the function $q$ being a $\text{P}_{\text{I}}$ transcendent. 
\end{proposition}

\begin{proof}
We perform a generalised version of the blow-up from the chart $(y_5,z_5)$ at the point $(q,q')$, obtaining the two charts $(u,v),(U,V)$. We select the chart $(U,V)$, i.e.
\begin{equation}\label{eq:blowup_syst5_gen}
\begin{cases}
    y_5 =  a_1\,UV + q \,,  \\[2ex]
    z_5 =  a_2\,V + q' \,,
\end{cases}
\end{equation}
and perform a further change of variables mimicking~\eqref{eq:inversion}, i.e.\ $(U,V)\mapsto(y_{92\text{m}}=U^{-1},z_{92\text{m}}=UV)$. This leads to the system in $(y_{92\text{m}},z_{92\text{m}})$ in~\eqref{eq:system_92_mod}, that undergoes the regularisation process that here we show. Starting from the identification of the base points in~$\mathbb{P}^2$: 
\begin{equation}
    b^{(1)}_{92\text{m}}\colon (u_0,v_0) = (0\,,0)\,, \qquad b^{(2)}_{92\text{m}}\colon (U_0,V_0) = (0\,,0)\,,
\end{equation}
from the base point $b^{(1)}_{92\text{m}}$, the chain of blow-ups to be performed is: 
\begin{equation}
    b^{(1)}_{92\text{m}} \colon (u_0,v_0) = (0\,,0)~\leftarrow ~ (u_1^{(1)},v_1^{(1)}) = (0\,,0)\,.
\end{equation}
The system in both charts $(u_2^{(1)},v_2^{(1)})$ and $(U_2^{(1)},V_2^{(1)})$ is regular, without imposing any additional condition. 
From the base point $b^{(2)}_{92\text{m}}$, the chain of transformations is: 
\begin{equation}
    \begin{aligned}
        &b^{(2)}_{92\text{m}}\colon (U_0,V_0) = (0\,,0) ~\leftarrow~ (u_1^{(2)},v_1^{(2)}) = (0\,,0) ~\leftarrow~ (u_2^{(2)},v_2^{(2)}) = \left(0\,,\frac{a_2^2}{4\,a_1^3}\right) ~\leftarrow~ \\[1ex]  
        &~\leftarrow~ (u_3^{(2)},v_3^{(2)}) = (0\,,0) ~\leftarrow~ (u_4^{(2)},v_4^{(2)}) = \left(0\,,-\frac{3\,a_2^4}{16\,a_1^7}\,q\right)~\leftarrow~ \\[1ex]
        &~\leftarrow~(u_5^{(2)},v_5^{(2)}) = \left(0\,,-\frac{a_2^5}{32\,a_1^9}\,q'\right) ~\leftarrow~(u_6^{(2)},v_6^{(2)}) = \left(0\,,-\frac{a_2^6}{128\,a_1^{11}}\,\big(q''- 36\,q^2 \big)\right) ~\leftarrow~ \\[1ex]
        &~\leftarrow~ (u_7^{(2)},v_7^{(2)}) = \left(0\,,\frac{a_2^7}{256\,a_1^{13}}\,\big( q'''-36\,q\,q' \big)\right).
    \end{aligned}
\end{equation}
The system in coordinates $\big(u_8^{(2)},v_8^{(2)}\big)$ and $\big(U_8^{(2)},V_8^{(2)}\big)$ is regular if the following condition is satisfied
\begin{equation}\label{eq:condP1}
    q^{(IV)} - 12\big( (q')^2 +q\,q'' \big) =0 \,,
\end{equation}
that can be integrated twice in $x$, obtaining 
\begin{equation}\label{eq:cond_p1}
    q'' - 6 q^2= c_1\,x+c_0\,, 
\end{equation}
where $c_0,c_1 \in \mathbb{C}$ are integration constants. 
\end{proof} 

\begin{remark}
The expression in~\eqref{eq:cond_p1} includes the elliptic case ($c_1=0$) that we do not consider here, and the $\text{P}_{\text{I}}$ transcendent ($c_1 \neq 0$ after a suitable re-scaling).       
\end{remark}

The construction we propose can be iterated. For instance, we can start from the system IX.B(2)~\eqref{eq:system_92} and perform a blow-up from the chart $(y_{92},z_{92})$ at a point of the solution curve, i.e. 
\begin{equation}
    p_{92}\colon (y_{92},z_{92}) = \big( q_1 ,  q_1^2-12 q+q_1' \big) \,, 
\end{equation}
where $q_1$ is solution to the second order differential equation for $y_{92}$ given by 
\begin{equation}\label{eq:second_order_y92}
    y_{92}'' =  12\, q'-12 q y_{92}-y_{92} y_{92}'+y_{92}^3\,. 
\end{equation}
We consider the transformations 
\begin{equation}
    \begin{cases}
    u=y_{92} - q_1\,, \\[2ex] 
    v=\dfrac{z_{92}-(q_1'+q_1-12q^2)}{y_{92}-q_1}\,,    
    \end{cases}
     \quad \begin{cases} 
    U= \dfrac{y_{92} - q_1}{z_{92} - (q_1'+q_1-12q^2)}  \,, \\[2ex] 
    V = z_{92} - (q_1'+q_1-12q^2)\,, \end{cases}
\end{equation}
and selecting the chart $(U,V)$, the further change of variables is 
\begin{equation}\label{change}
    (U,V) \mapsto \big( y_{92}^{(1)}=U^{-1} \,, z_{92}^{(1)} = UV \big)\,. 
\end{equation}
The iteration of the procedure results in the system in $\big(y_{92}^{(1)},z_{92}^{(1)}\big)$ 
\begin{equation}\label{eq:iterated_bureau}
    \begin{cases}
        \big(y_{92}^{(1)}\big)'=
        3 y_{92}^{(1)} q_1+q_1'+q_1{}^2-12\, q+2 y_{92}^{(1)} z_{92}^{(1)}-\big(y_{92}^{(1)}\big)^2
        , \\[1ex]
        \big(z_{92}^{(1)}\big)'=
       z_{92}^{(1)} \left( y_{92}^{(1)}-z_{92}^{(1)}-2\,  q_1\right)
\,,
    \end{cases}
\end{equation}
where $q(x)$ is a $\text{P}_{\text{I}}$ transcendent, and $q_1(x)$ is a solution of~\eqref{eq:second_order_y92}. This system can be considered as a  generalised version of the Bureau-Guillot system. 
Indeed, it is birationally equivalent to the Hamiltonian system V related to $(\text{P}_{\text{I}})$ via the transformation 
\begin{equation}
     y_{92}^{(1)}=\frac{(y_{5}-q) \left(11 q-q_{1}'-q_{1}^2+y_{5}\right)}{z_{5}-q'+q_{1} (q-y_{5})} \,, \qquad 
     z_{92}^{(1)}= \frac{z_{5}-q'+q_{1} (q-y_{5})}{y_{5}-q} \,,
\end{equation}
given by the composition of four birational transformations. In addition, the system~\eqref{eq:iterated_bureau} presents meromorphic coefficients which are solutions to $(\text{P}_{\text{I}})$ and the system IX.B(2) and their derivatives. Finally, the regularising conditions are equations with polynomial functions of $q(x)$, $q_1(x)$ and their derivatives 
\begin{equation}
    F\big(q,q',q'',q_1, q_1', q_1'';x\big) = 0\,, 
\end{equation}
which are identically satisfied once we assume that $q(x)$ is the $\text{P}_{\text{I}}$ transcendent, and $q_1(x)$ solves the second order equation for the system IX.B(2), i.e.~\eqref{eq:second_order_y92}.
Note that the similar regularising  conditions hold if we omit the change of variables (\ref{change}). 
Thus, we conjecture that we can further iterate the method applying it to system \eqref{eq:iterated_bureau} and so on, to obtain the whole orbit of the Bureau-Guillot systems.

We end this Section with the following remarks for the system with arbitrary coefficients and a system of quasi-Painlev\'e type\footnote{A quasi-Painlev\'e system has the so-called quasi-Painlev\'e property \cite{Shimomura2006,Shimomura2008}: the admissible movable singularities are algebraic branch points rather than poles. }. 
\begin{remark}
We consider the system V in~\eqref{eq:system_5} as written with the arbitrary coefficient function $f(x)$ 
\begin{equation}\label{eq:systV_arbi_coeff}
    \begin{cases}
    y_5' = z_5\,,\\[1ex]
z_5' = 6\,y_5^2 + f(x)\,, 
\end{cases}
\end{equation}
and we perform again the blow-up at the point $p_5\colon (y_5,z_5)=(q,q')$ as in~\eqref{eq:point_syst_5}. Therefore, we consider the transformation $(y_5,z_5) \mapsto (U,V)$ as in~\eqref{eq:blowup_syst5} with $a_1=a_2=1$, and the additional change of variables~$(U,V) \mapsto (U^{-1},UV)$. We obtain the modified version of the system IX.B(2) as in~\eqref{eq:system_92_mod}, with~$a_1=a_2=1$. In this case, we find the same condition as in~\eqref{eq:condP1},  which can be integrated twice and gives us that $q$ satisfies $(\text{P}_{\text{I}})$. As a consequence, since $q(x)$ is solution to the second order equation associated with the system~\eqref{eq:systV_arbi_coeff} as well, the arbitrary coefficient $f(x)$ is linear in $x$. 

\end{remark}

\begin{remark}\label{rmk:thomas}
We consider the equation for the quasi-Painlev\'e I equation~\cite{Shimomura2006}, i.e. 
\begin{equation}\label{eq:qsiP1}
(\text{qsi-P}_{\text{I}})\colon ~~     (y_1)'' =  k_1(y_1)^4 + x\,,
\end{equation}
and the corresponding system
\begin{equation}\label{eq:qsiP1_syst}
    \begin{cases}
    y_1' = z_1\,,\\[1ex]
z_1' = k_1\,y_1^4 + x\,,
\end{cases}
\end{equation}
with $k_1 \in \mathbb{C}$ a constant. 
By performing the blow-ups as described in this Section, we are not able to produce an analogue of the Bureau-Guillot system for this case. In particular, the quasi-Painlev\'e equation~\eqref{eq:qsiP1} does not appear as a regularising condition for the system, and probably some changes in the construction are needed for the procedure to work in this case. For instance, in Section~\ref{sec:mixedcases} we present an alternative simpler version of the change of variables, that allows us to produce a system where the quasi-Painlev\'e equation is indeed found as a regularising condition (Example~\ref{ex:quasiP1coeff}). 
\end{remark}

\subsection{\texorpdfstring{Systems IX.B(3) and (parametrised) XIII for $(\text{P}_{\text{II}})$}{P2}}

Here we consider the systems IX.B(3) and  XIII following the nomenclature in~\cite{Bureau,Guillot}. In the variables~$(y_{93},z_{93})$ the system IX.B(3) has the form 
\begin{equation}\label{eq:syst_93}
\begin{cases} 
y_{93}' = z_{93} -y_{93}^2 - \dfrac{f(x)}{2}, \\[2ex]
z_{93}' = 2\,y_{93}\,z_{93} + \alpha+\dfrac{1}{2},
\end{cases} \qquad f''(x) = 0\,, \quad \alpha\in \mathbb{C} \text{ is a  constant}\,.
\end{equation}
In~\cite{MDGF}, we show how the condition $f'' = 0$ can be obtained as the regularising condition for the system~IX.B(3) to have the Painlev\'e property and how the system is  related to $(\text{P}_{\text{II}})$ 
\begin{equation}\label{eq:P2}
    (\text{P}_{\text{II}})\colon ~y'' = 2y^3 +x\,y+\alpha\,,
\end{equation}
via the Okamoto-Sakai approach, resulting in the surface type $E_7^{(1)}$. 
The system XIII in the variables $(y_{13},z_{13})$ is 
\begin{equation}\label{eq:syst_13}
\begin{array}{c c}
\begin{cases} 
y_{13}' = \dfrac{y_{13}}{2}\,(2\,z_{13} - y_{13}) + 2\,p(x)\,y_{13} ,\\[3ex]
z_{13}' = \dfrac{z_{13}}{2}(3\,y_{13} - 2\,z_{13}) - 4\,p(x)\,z_{13} + 2\,p(x)^2 - 2\,p'(x) + f(x)\,,
\end{cases}   
\end{array}
\end{equation}
with the conditions on the coefficient functions $p(x)$ and $f(x)$ 
\begin{equation}
\begin{aligned}
    (p'' -  2\,p^3 - f\,p)' &= 0, \\[1ex]
      f'' &= 0\,. 
\end{aligned}
\end{equation} 
The system XIII is related to $(\text{P}_{\text{II}})$ if the coefficient function $p(x)$ is itself a $\text{P}_{\text{II}}$ transcendent. In~\cite{MDGF}, we determined the latter condition as the regularising condition for the system, and via the Okamoto-Sakai approach we established the birational transformation between the two systems.

    We find a modified version of the system XIII from blowing up the chart $(y_{13},z_{13})$ at the point of the orbit of the system IX.B(3), i.e.
    \begin{equation}\label{eq:p93_blowup}
        p_{93}\colon (y_{93},z_{93}) = \left( q, \, q' + q^2 +\frac{x}{2} \right)\,,
    \end{equation}
    introducing the constant parameters $a_1,a_2 \in \mathbb{C}$. 

    \begin{proposition}
    The modified system XIII  depending on $a_1,a_2 \in \mathbb{C}$ given by 
        \begin{equation}\label{eq:syst_13_mod}
    \begin{cases}
        (y_{13\text{m}})'=\dfrac{a_1}{a_2} \left(2q^2+2 q'+x\right)-\dfrac{a_2 }{a_1}\,(y_{13\text{m}})^2+3\, a_1 y_{13\text{m}} \,z_{13\text{m}}+4\, y_{13\text{m}} q\,,\\[2ex]
    (z_{13\text{m}})'=z_{13\text{m}} \left(\dfrac{a_2}{a_1}\,y_{13\text{m}}-a_1 z_{13\text{m}}-2 q\right),
        \end{cases}
    \end{equation}
    with the function $q$ being a $\text{P}_{\text{II}}$ transcendent is a Bureau-Guillot system in the sense of Definition~\ref{def:BG}. 
    \end{proposition}

\begin{proof}
We consider a generalised blow-up at the point $p_{93}$ in~\eqref{eq:p93_blowup}, selecting the chart $(U,V)$
    \begin{equation}
    \begin{cases}
            y_{93} = a_1\, UV +q\,,  \\[2ex]
            z_{93} = a_2\,V + q' + q +\dfrac{x}{2}\,.
        \end{cases}
    \end{equation}
    With the further change of variables $(U,V)\mapsto (U^{-1},UV)$, we get the modified version of the system~XIII~\eqref{eq:syst_13}, which we recover by taking $a_1=a_2=1/2$ and  $q =-p$. 
    The condition for the system~\eqref{eq:syst_13_mod} to be regularised is 
    \begin{equation}\label{eq:condP2}
        q'''= 6q^2\, q' +q +  xq'    \,.  
    \end{equation}
    This can be easily integrated once after being written as a total derivative as 
    \begin{equation}
        \big( q''-2q^3-xq \big)' = 0  \quad \implies \quad q''-2q^3-xq = \alpha\,,
    \end{equation}
    hence $q(x)$ solves \hyperref[eq:P2]{$(\text{P}_{\text{II}})$}. 
   
\end{proof}

\section{\texorpdfstring{Analogues of Bureau-Guillot systems}{analogues}}\label{sec:analogueBG}
{
We present the detailed construction (and regularisation) of the Bureau-Guillot systems associated with the third Painlev\'e equation of type $D_6$, $D_7$ and $D_8$, and of the fourth Painlev\'e equation. For the cases of the fifth and sixth Painlev\'e equations, we describe the framework to obtain the corresponding Bureau-Guillot systems, which are reported in the Appendices~\ref{app:P5} and~\ref{app:P6} respectively. For all the cases, it was essential to follow the geometric approach, i.e.\ graphically building the rational surface from which we recognise the Dynkin diagram uniquely identifying the specific Painlev\'e equation (see e.g.~\cite{Dzhamay2109.06428}). 

\subsection{General construction}\label{sec:generalconstruction}
Based on the observations reported in the previous section for the known cases of $(\text{P}_{\text{I}})$ and $(\text{P}_{\text{II}})$, we outline the general method that we use in the following. Below, we consider the Hamiltonian systems associated with the Okamoto polynomial Hamiltonians $H_{\text{J}}(y,z;x)$ for the Painlev\'e equations~$(\text{P}_{\text{J}})$ with $\text{J}=\text{III-}6,\,\text{III-}7,\,\text{III-}8, \,\text{IV},\,\text{V}, \text{VI}$, referring to~\cite{okamoto1,okamoto2,okamoto3,Ohyama}. 

The system in the variables $\big(y,z\big)$ associated with a certain $(\text{P}_{\text{J}})$ allows us to determine the orbit~$\mathcal{O}_{\text{J}}$ defined as  
\begin{equation}
    \mathcal{O}_{\text{J}} = \{ (y,z) \in \mathbb{A}^2_{(y,z)}  \,|\,  z = h_\text{J}\big( y, y';x \big)  \}\,,
\end{equation}
where $y$ is solution of $(\textup{P}_\textup{J})$. We address the affine chart as $(y_{\text{J}},z_{\text{J}})$, and we consider the following (non-autonomous) system of ODEs related to the Painlev\'e equation~$(\text{P}_{\text{J}})$ 
\begin{equation}\label{eq:gen_syst}
    \begin{cases}
        y_\text{J}'= f_\text{J}\big(y_\text{J},z_\text{J};x\big) \,,\\[1ex]
        z_\text{J}'= g_\text{J}\big(y_\text{J},z_\text{J};x\big)  \,.
    \end{cases}
\end{equation}
In particular, we assume that $f_{\text{J}}$ in the first equation is linear in $z_\text{J}$, and specifically
$z_\text{J} = h_\text{J} \big( y_\text{J},y_\text{J}';x \big)$.
Substituting this expression for $z_\text{J}$ and its derivative $z_\text{J}'$ into the second equation of~\eqref{eq:gen_syst}, we have the second order equation in $y_\text{J} \equiv y$:  
\begin{equation}\label{eq:secord_gen}
    y'' = \frac{1}{(h_\text{J})_{y'}} \left(\, g_\text{J}\!\left( y,h_\text{J}\big(y,y';x\big);x \right)-(h_\text{J})_{y}\big(y,y';x\big) y' - (h_\text{J})_x(y,y';x) \,\right) , 
\end{equation}
with $(h_\text{J})_x=\partial h_\text{J} /\partial x$, $(h_\text{J})_y=\partial h_\text{J} /\partial y$, and $(h_\text{J})_{y'}=\partial h_\text{J} /\partial y'$.

We introduce $q(x)$ as a solution of $(\textup{P}_\textup{J})$, therefore a generic point $p_\text{J}$ on the orbit $\mathcal{O}_{\text{J}}$ in the chart~$(y_\text{J},z_\text{J})$ is given by 
\begin{equation}\label{eq:general_orbit}
    p_\text{J} \colon (y_\text{J},z_\text{J}) = \big( q \,,  h_{\text{J}}(q,q';x)\big) \,. 
\end{equation}
We construct systems which are Bureau-Guillot according to Definition~\ref{def:BG} by considering one of the following changes of variables $(y_\text{J},z_\text{J}) \mapsto (y,z)$, i.e.\ the birational transformations
\begin{equation}
\begin{array}{l l l}
\hspace{5ex} & {\label{eq:tran1}(a)~~ \begin{cases}
    y_\text{J}= y + q \,,    \\[1ex] 
    z_\text{J}= yz + h_{\text{J}}(q,q';x) \,,   
\end{cases} }  &  
(b)~~ \begin{cases}
    y_\text{J}= z + q \,,    \\[1ex]
    z_\text{J}= yz + h_{\text{J}}(q,q';x) \,,  
\end{cases}  \\[6ex]
& (c)~~\begin{cases}
    y_\text{J}= yz + q \,,    \\[1ex]
    z_\text{J}= y + h_{\text{J}}(q,q';x) \,,
\end{cases} & 
(d)~~\begin{cases}
    y_\text{J}= yz + q \,,    \\[1ex]
    z_\text{J} = z + h_{\text{J}}(q,q';x) \,.
\end{cases} 
\end{array}
\end{equation} 

\vspace{2ex}

\noindent
We call these changes of variables transformations of type I (or blow-up-like).
After applying any of the above transformations, the resulting system in $(y,z)$ is still related to the equation~$(\textup{P}_{\textup{J}})$, with $q(x)$ being solution to~$(\textup{P}_{\textup{J}})$ itself appearing as a regularising condition. It is easy to see that the composition of the two changes of coordinates used in the previous section corresponds to the transformation~\hyperref[eq:tran1]{$(b)$}. 

Applying different type I transformations results in systems with structurally distinct blow-up cascades. 
However, they 
 share similar regularising conditions, i.e.\ the function $q(x)$ is itself a solution of~$(\text{P}_{\text{J}})$. Throughout this section, we build the Bureau-Guillot systems adopting the transformation~\hyperref[eq:tran1]{$(b)$} to be consistent with the previous section.

\subsection{\texorpdfstring{Systems for $(\text{P}_{\text{III}},D_6)$, $(\text{P}_{\text{III}},D_7)$ and $(\text{P}_{\text{III}},D_8)$}{p3}}
The Painlev\'e III equation  in its general form is expressed as 
\begin{equation}\label{eq:P3_D6}
    (\text{P}_{\text{III}},D_6) \colon ~~ y''= \dfrac{(y')^2}{y}-\dfrac{y'}{x}+ \dfrac{\alpha\,y^2+\beta}{x} +\gamma\,y^3+\dfrac{\delta}{y} \,, \qquad y \in \mathbb{C} \,, ~~\gamma \delta \neq 0\,,
\end{equation}
with $\alpha, \beta, \gamma, \delta \in \mathbb{C}$ constant parameters. As specified in~\cite{Ohyama}, depending on the values of the parameters, we can identify two other standard forms of the third Painlev\'e equation, which for historical reasons are called ``of special type''. The three versions of the equation are related to a certain surface type, as usual corresponding to the configuration of vertical leaves to be removed from the compactified projective space. In particular, the equation~\eqref{eq:P3_D6} has a surface type represented by the extended Dynkin diagram~$D^{(1)}_6$, the special types have surface types~$D^{(1)}_7$ and~$D^{(1)}_8$ respectively. In terms of the parameters, this is 
\begin{align}
    (\text{P}_{\text{III}},D_7)&\colon \eqref{eq:P3_D6} \text{ with } \gamma= 0\,, \alpha \delta \neq 0 \text{ or with }  \delta=0\,, \beta \gamma \neq 0 \,,
    \label{eq:P3_D7}\\[1ex]
    (\text{P}_{\text{III}},D_8)&\colon \eqref{eq:P3_D6} \text{ with } \gamma= 0\,, \delta =0 \,, \alpha \beta \neq 0 \,.  
    \label{eq:P3_D8}
\end{align}

To construct the system with the polynomial right-hand side and with the  Painlev\'e transcendent in their coefficient functions, we start with  the Okamoto Hamiltonian for~{\hyperref[eq:P3_D6]{$(\text{P}_{\text{III}},D_6)$}} as in~\cite{okamoto4}
\begin{equation}\label{eq:H3ok}
H_{\text{III-}6}^{\text{Ok}}\big(y_3,z_3;x\big)= \dfrac{1}{x}\big(2\, y_3^2 \,z_3^2 +\eta_{\infty}  (\kappa_{0}+\kappa_{\infty}) x \,y_3-z_3 \left( 2 \,\eta_{\infty} x \,y_3^2  - 2\eta_{0}\, x+(2 \kappa_{0}+1) y_3\right) \big) \,,
\end{equation}
and the corresponding  Hamiltonian system  \begin{equation}
    \dfrac{dy_3}{dx}=\dfrac{\partial H^{\text{Ok}}_{\text{III}}}{\partial z_3},\qquad \dfrac{dz_3}{dx}=-\dfrac{\partial H^{\text{Ok}}_{\text{III}}}{\partial y_3},
\end{equation} 
which is equivalent to~\eqref{eq:P3_D6} 
with parameters $\eta_0, \eta_{\infty}, \kappa_0, \kappa_\infty$  related to $\alpha, \beta, \gamma, \delta$ as 
\begin{equation}
    \alpha = -4\eta_{\infty}\kappa_{\infty}\,, \quad \beta = 4\eta_0(\kappa_0+1)\,, \quad \gamma = 4\eta_{\infty}^2 \,, \quad \delta = -4\eta_0^2\,. 
\end{equation}
Of the two variables, $y_3$ satisfies the equation~\eqref{eq:P3_D6}, i.e.  
\begin{equation}\label{eq:P3_D6_eta_kappa}
    y_3'' = \frac{(y_3')^2}{y_3}-\frac{y_3'}{x}+\frac{4 \eta_{0} (\kappa_{0}+1)-4 \eta_{\infty} \kappa_{\infty} y_3^2}{x}+4 \eta_{\infty}^2 y_3^3-\frac{4 \eta_{0}^2}{y_3}\,. 
\end{equation}

We have the following 
\begin{proposition}\label{thm:P3D6} Applying the transformation \hyperref[eq:tran1]{$(b)$} of type I to the Hamiltonian system with Hamiltonian $H_{\textup{III-}6}^{\textup{Ok}}$~in \eqref{eq:H3ok} for the chart $(y_{3},z_{3})$ along the orbit $\mathcal{O}_{\text{III-6}}$, i.e.
\begin{equation}\label{eq:point3Ok}
    p_{\textup{III-}6}^{\textup{Ok}} \colon (y_3,\,z_3) = \left( q\,,\dfrac{x\,\eta_{\infty}}{2}+\dfrac{x(q'-2  \eta_{0})}{4 q^2}+\dfrac{2 \kappa_{0}+1}{4 q}\right)\,,
\end{equation}
gives rise to the system 
\begin{equation}\label{eq:syst_P3}
        \begin{cases}
        \begin{aligned} 
            y'&= \left( 3\,\dfrac{yz}{q^2}+ 4\,\dfrac{y}{q}+ \dfrac{1+2\kappa_0}{2q^3} \right) (2\eta_0-q') + x \left(\eta_{\infty}^2 - \dfrac{\eta_0^2}{q^4} \right) +\dfrac{x q'}{4 q^4} (4 \eta_0 - q') \\[1ex]
            &~~ - \!\left( 3\,\dfrac{yz}{x q} + 2\,\dfrac{y}{x} \right) (1+2\kappa_0) - \dfrac{(1+2\kappa_0)^2}{4x q^2} -\dfrac{4 y^2}{x}\left(3 z q+q^2+2 z^2\right), \\[2ex]\end{aligned}  \\
            z'= \dfrac{z}{x} \left( \left( \dfrac{z}{q}+1  \right)(1+2\kappa_0)+4y(z+q)^2 \right) - \left( \dfrac{z^2}{q^2}+ \dfrac{2z}{q} \right)(2\eta_0 - q') ,
        \end{cases}
    \end{equation}
    which is Bureau-Guillot in the sense of Definition~\ref{def:BG} for $(\textup{P}_{\textup{III}},D_6)$ if $q$ is in turn solution to~$(\textup{P}_{\textup{III}},D_6)$.
\end{proposition}
\begin{proof} 
We perform the change of variables inverting~\hyperref[eq:tran1]{$(b)$} from $(y_3,z_3)\mapsto (y,z)$ along the orbit $\mathcal{O}_{\text{III-}6}$, i.e.
\begin{equation}
\begin{cases}
    z = y_3-q \,, \\[2ex]
    y = \dfrac{z_3-\left(\dfrac{x\,\eta_{\infty}}{2}+\dfrac{x(q'-2  \eta_{0})}{4 q^2}+\dfrac{2 \kappa_{0}+1}{4 q}\right)}{y_3-q}\,. 
\end{cases}
\end{equation} 
The Hamiltonian system associated with~\eqref{eq:H3ok} is transformed into the system~\eqref{eq:syst_P3}. 

We now consider the extended phase space $\mathbb{P}^2$, and identify the base points and the following cascades of blow-ups. From the base point $b_1$ coinciding with the origin in the affine chart $(u_0,v_0)$ we find two cascades that we distinguish with $b_1^{(1)}$ and $b_1^{(2)}$ 
\begin{align}
\label{eq:P3D6_casc_b1}
    b_1^{(1)}&\colon(u_0\,,v_0) = (0\,,0) ~ \leftarrow ~ \big(u_1^{(1)},v_1^{(1)}\big) = (0\,,0)\,, \\[1ex]
    \label{eq:P3D6_casc_b2}
    \begin{split} 
    b_1^{(2)}&\colon(u_0\,,v_0) = (0\,,0) ~ \leftarrow ~ \big(u_1^{(2)},v_1^{(2)}\big) = (0\,,-q) ~ \leftarrow ~ \big(U_2^{(2)},V_2^{(2)}\big)  = \left( 0\,, 0 \right) ~ \leftarrow ~ \\[1ex]
    &~ \leftarrow ~ \big(\widehat{U}_3^{(2)},\widehat{V}_3^{(2)}\big) = \left(\dfrac{\eta_0\,x}{q}\,,0\right) ~ \leftarrow ~ \big(U_4^{(2)},V_4^{(2)}\big) = \left( \frac{ \eta_{0}\,x-(\kappa_{0}+1) q}{q^2}\,, 0\right),
    \end{split} 
\end{align}
with the intermediate change of variables $\big(U_3^{(2)}, V_3^{(2)} \big) \mapsto \big((\widehat{U}_3^{(2)})^{-1}, \widehat{V}_3^{(2)} \big)$, following~\cite{GF_AS1,GF_AS2}.
The resulting systems in the final charts, respectively in coordinates $(u_2^{(1)},v_2^{(1)})$ and $(u_5^{(2)},v_5^{(2)})$, are regular without imposing further conditions on the coefficient functions. From the further base point $b_2$ coinciding with the origin in the affine chart $(U_0, V_0)$, we obtain the following two cascades,  identified with $b_2^{(\pm)}$:
\begin{align}
\begin{aligned} 
    &b_2^{(\pm)} \colon (U_0\,,V_0) = (0\,,0) ~ \leftarrow ~ \big(U_1^{(\pm)},V_1^{(\pm)}\big) = (0\,,0) ~ \leftarrow ~ \\[1ex]
    &~ \leftarrow ~ \big(U_2^{(\pm)},V_2^{(\pm)}\big) = \left(\dfrac{x }{4q^2}\left(2(\eta_{0} \pm \,\eta_{\infty}q^2) -q'\right)-\dfrac{2 \kappa_{0}+1}{4 q}\,,0\right) \\[1ex]
    &~ \leftarrow ~ \big(U_3^{(\pm)},V_3^{(\pm)}\big) = \left(\dfrac{\kappa_0\mp \eta_{\infty}q\,x}{2}\pm \dfrac{x}{8 \eta_{\infty} }  \left(\dfrac{q''-4 \eta_{0} (\kappa_{0}+1)}{ q^2} +\dfrac{4 \eta_{0}^2-(q')^2}{ q^3}\right)\pm \dfrac{q'}{8\eta_{\infty}q^2}\,,0\right)
\end{aligned} 
\end{align}
and the systems in the final charts $(U_4^{(\pm)},V_4^{\pm})$  are regularised by the following condition: 
\begin{equation}\label{eq:condP3D6}
    x q''' q^2-4q' \left(3 x \eta_{0}^2-2q \eta_0 (\kappa_{0}+1)+x \eta_{i}^2 q^4\right)-3 (q')^2(q + x q')-2 q q''(2x q'-q ) +4 q( \eta_{0}^2 - \eta_{i}^2 q^4) = 0\,, 
\end{equation}
 satisfied by $q(x)$ being solution to {\hyperref[eq:P3_D6]{$(\text{P}_{\text{III}},D_6)$}} as in~\eqref{eq:P3_D6_eta_kappa}. 
\end{proof} 

\begin{remark}
    The conditions we found for the case of {\hyperref[eq:P1]{$(\text{P}_{\text{I}})$}} and {\hyperref[eq:P2]{$(\text{P}_{\text{II}})$}} (respectively in equations~\eqref{eq:condP1} and~\eqref{eq:condP2}) can be easily integrated, showing that they appear as first or second order derivatives of Painlev\'e equations. 
    In the present case and in the forthcoming cases, we only show that {\hyperref[eq:P3_D6]{$(\text{P}_{\text{III}},D_6)$}} solves the regularising condition~\eqref{eq:condP3D6}, without explicitly integrating it, given the complexity of the expressions. 
\end{remark}

We analyse both versions of $(P_{\text{III}},D_7)$ in~\eqref{eq:P3_D7}, starting from the case $\gamma=0$ and $\alpha\delta \neq 0$. The corresponding Hamiltonian is 
\begin{equation}\label{eq:H3D7}
    H^{\text{Ok1}}_{\text{III-}7}\big( y_3,z_3;x \big) = \frac{1}{x} \left(2 y_3^2 z_3^2+\kappa_{\infty} x y_3-z_3 ((2 \kappa_{0}+1) y_3-2 \eta_{0} x) \right), 
\end{equation}
where the parameters $\eta_0,\kappa_0, \kappa_{\infty}$ are related to $\alpha,\beta,\delta$ by \begin{equation}
    \alpha =-4 \kappa_{\infty},\qquad \beta =4 \eta_{0} (\kappa_{0}+1),\qquad \delta =-4 \eta_{0}^2\,. 
\end{equation}
The variable $y_3$ satisfies the equation $(P_{\text{III}},D_7)$ 
\begin{equation}
	y_3''= \frac{(y_3')^2}{y_3}-\frac{y_3'}{x}+\frac{4}{x}\big( \eta_{0}(\kappa_{0}+1)-\kappa_{\infty} y_3^2\,\big)-\frac{4 \eta_{0}^2}{y_3}\,.
\end{equation}
\begin{proposition}
Applying the transformation \hyperref[eq:tran1]{$(b)$} of type I to the Hamiltonian system with Hamiltonian $H_{\textup{III-}7}^{\textup{Ok1}}$~in \eqref{eq:H3D7} for the chart $(y_{3},z_{3})$ along the orbit $\mathcal{O}_{\text{III-7}}$, i.e.
\begin{equation}\label{eq:p3D7-1}
    p_{\textup{III-}7}^{\textup{Ok}1}\colon (y_3,z_3) = \left( q\,, \dfrac{x(q'-2  \eta_{0})}{4 q^2}+\dfrac{2 \kappa_{0}+1}{4 q} \right) \,,
\end{equation}
gives rise to the system 
\begin{equation}\label{eq:BG_P3_D7}
    \begin{cases}
	    \begin{aligned} 
        y' &= 2y\left(\frac{2\eta_{0}-q'}{q} -\frac{2 y }{x} \left(q^2+3qz+2z^2\right) \right)  -\left(\frac{2 \kappa_{0}+1}{2 \sqrt{x} q}-\frac{\sqrt{x} \left(2 \eta_{0}-q'\right)}{2 q^2}\right)^{\!\!2}\\[1ex]
        &~~+y \left(\frac{3 z}{q}+2\right) \left(\frac{2 \eta_{0}-q'}{q}-\frac{2 \kappa_{0}+1}{x}\right),     
        \end{aligned} \\[8ex]
        \begin{aligned}
            z' &= \left(\frac{z^2}{q}+z\right) \!\left(\frac{4 y q(q+z) }{x}+\frac{2 \kappa_{0}+1}{x}\right) -\left(\frac{z^2}{q}+2 z\right)\! \frac{2 \eta_{0}-q'}{q}\,,
            \end{aligned}
	\end{cases}
\end{equation} 
which is Bureau-Guillot in the sense of Definition~\ref{def:BG} for $(\textup{P}_{\textup{III}},D_7)$ if $q(x)$ is solution to $(\textup{P}_{\textup{III}},D_7)$ with $\gamma=0$. 
\end{proposition}

\begin{proof}
The system in the coordinates $(y,z)$ is obtained by performing the transformation~\hyperref[eq:tran1]{$(b)$} of type I with~$(y_3,z_3)\mapsto (y,z)$ at the point $p_{\textup{III-}7}^{\textup{Ok}1}$ in~\eqref{eq:p3D7-1} of the orbit $\mathcal{O}_{\text{III-}7}$. 

    We then proceed with the compactification of the complex plane $(y,z)$ to $\mathbb{P}^2$, and analyse the Hamiltonian system associated with $H^{\text{Ok}1}_{\text{III-}7}$. From the base point $b_1$ two branches emerge, with the same chains of blow-ups reported in~\eqref{eq:P3D6_casc_b1} and~\eqref{eq:P3D6_casc_b2} for the case {\hyperref[eq:P3_D6]{$(\text{P}_{\text{III}},D_6)$}}. 
    From the point $b_2$, we have a single cascade with blow-ups
    \begin{align}
\begin{aligned} 
    &b_2^{(3)} \colon (U_0\,,V_0) = (0\,,0) ~ \leftarrow ~ \big(U_1^{(3)},V_1^{(3)}\big) = (0\,,0) ~ \leftarrow ~ \\[.2ex]
    &~ \leftarrow ~ \big(U_2^{(3)},V_2^{(3)}\big) = \left(\dfrac{x }{4q^2}\left(2\eta_{0} -q'\right)-\dfrac{2 \kappa_{0}+1}{4 q}\,,0\right)\\[.2ex]
    &~\leftarrow ~ \big(u_3^{(3)},v_3^{(3)}\big) = (0\,,0) ~\leftarrow ~  \big(U_4^{(3)},V_4^{(3)}\big) = \left( \frac{x^2}{2 q^3} (4\eta_{0}^2 - (q')^2) + \frac{x q''+q'-4 \eta_{0} (\kappa_{0}+1)}{8 q^2} \,, 0 \right)  \\[.2ex]
    &~ \leftarrow ~ \big(U_5^{(3)},V_5^{(3)}\big) =  \left( \frac{x}{8 q^2}\left(q \kappa_{0} \left(x q''+q'\right)-4 \eta_{0} (\kappa_{0}+1) \left(x q'+\kappa_{0} q\right)\right)+ \right. \\[.2ex]
    &\hspace{5ex}\left. + \frac{x^2}{16 q^4} \left(  3 x q'(4 \eta_{0}^2-(q')^2) +q\left(q' \left(4 x q''+3 q'\right)-4\left(4 x q'''+q q''+2 \eta_{0}^2\right)+2  \kappa_{0} \right) \right) \,,0 \right),
\end{aligned} 
\end{align}
where the upper index $(3)$ labelling the coordinates refers to the third cascade. 
The system in the coordinates $\big(U_6^{(3)},V_6^{(3)}\big)$ is regular if the following condition is satisfied 
\begin{equation*}
\begin{aligned} 
&\left(4 \eta_{0}^2-(q')^2\right) \Big(q \left(3 x^2 \left(4 x \eta_{0}^2 q''-(q')^2 \left(3 x q''+2 q'\right)\right)+4 \eta_{0} (\kappa_{0}+1) \left(6 x^2 (q')^2-x q \left(13 x q''+11 q'\right)\right)\right)\\
&~~~-\left(q^3 \left(x^3 q^{(IV)}+4 x^2 q'''+11 x q''+q'\right)\right)+9 x q^2 \left((q')^2+x q'' \left(x q''+2 q'\right)\right)\Big)\\
&+4 \eta_{0} (\kappa_{0}+1) q^2 \Big(q^2 \left(x^2 q^{(IV)}+4 x q'''+2 q''\right)-q \left(q' \left(2 x^2 q'''+5 q'\right)+2 x q'' \left(3 x q''+7 q'\right)\right)\\
&~~~+4 \eta_{0}^2 \left(4 x \left(2 x q''+q'\right)+q\right)\Big)+32 \eta_{0}^2 (\kappa_{0}+1)^2 q^2 \left(q \left(x q''+q'\right)-x (q')^2\right)+3 x^2 q' \left(q-x q'\right) \left(4 \eta_{0}^2-(q')^2\right)^2\\
&+2 q^2 q'' \left(-4 x^2 \eta_{0}^2 \left(x q''+4 q'\right)+x q \left(2 x^2 (q'')^2+3 x q' q''+24 \eta_{0}^2\right)-q^2 \left(q'-x q''\right)\right)\\
&-x^2 q^{(IV)} q^4 \left(x q''+q'\right)+x q''' q^3 \left(q \left(x^2 q'''-4 q'\right)+2 x \left(4 \eta_{0}^2-x q' q''\right)\right)-2 q^2 (q')^2 \left(4 x \eta_{0}^2-q q'\right)= 0\,,
\end{aligned} 
\end{equation*}
that is verified if $q(x)$ solves {\hyperref[eq:P3_D7]{$(\text{P}_{\text{III}},D_7)$}}.
\end{proof}

The second version of the Hamiltonian for {\hyperref[eq:P3_D7]{$(\text{P}_{\text{III}},D_7)$}} with $\delta=0, \alpha\beta \neq 0$ is 
\begin{equation}\label{eq:H3D72}
   H^{\text{Ok2}}_{\text{III-}7}\big( y,z;x \big)  = \frac{1}{x} \left( 2 y^2 z^2+\eta_{\infty} \kappa_{\infty} x y+\frac{\kappa_{0} x}{y}-z \left(2 \eta_{\infty} x y^2+y\right)\right)\,,
\end{equation}
with the parameters related by 
\begin{equation}
    \alpha = -4\kappa_{\infty}\eta_{\infty}\,, \qquad \beta = 4\kappa_0\,, \qquad \gamma = 4\eta_{\infty}^2\,.
\end{equation}
\begin{proposition}
Applying the transformation \hyperref[eq:tran1]{$(b)$} of type I to the Hamiltonian system with Hamiltonian $H_{\textup{III-}7}^{\textup{Ok2}}$~in \eqref{eq:H3D72} for the chart $(y_{3},z_{3})$ along the orbit $\mathcal{O}_{\text{III-7}}$, i.e.
\begin{equation}
    p_{\textup{III-}7}^{\textup{Ok}2}\colon (y,z) = \left( q\,,\dfrac{x\,\eta_{\infty}}{2}+\dfrac{x\,q'}{4 q^2}+\dfrac{1}{4 q}\right) \,,
\end{equation}
gives rise to the system 
\begin{equation}\label{eq:BG_P3_D72}
    \begin{cases}
	    \begin{aligned} 
        y' &= -\frac{1}{4 x (q+z)^2}\Bigg(\left(\frac{z^2}{q^2}+\frac{2 z}{q}+1\right) \left(\frac{x^2 (q')^2}{q^2}+\frac{2 x q'}{q}\right)+\frac{q'}{q}\left(\frac{3 z}{q}+10\right) \left(4 x y z^2 \right)\\[1ex]
        &~~~-4  \left(x^2 \eta_{\infty}^2(q+z)^2-2 y \left( q^2 \left(18 y z^2+1\right)-4 z^2 \left(y z^2+1\right)\right)+7 y z q \left(4 y z^2+1\right)\right)\\[1ex]
        &~~~+4 x y (4 q+11 z) q'-8 y^2 q^3 (2 q+10 z)+\frac{z}{q} \left(\frac{z}{q}+12 y z^2+2\right)+\frac{4 x \kappa_{0}}{q} \left(\frac{z}{q}+2\right)+1\Bigg),     
        \end{aligned} \\[8ex]
        \begin{aligned}
            z' &= \dfrac{q'}{q}\left(\dfrac{z^2}{q}+2 z\right)+\frac{4 y z}{x} (q+z)^2
            +\frac{z}{x} \left(\frac{z}{q}+1\right) \,,
            \end{aligned}
	\end{cases}
\end{equation} 
which is Bureau-Guillot in the sense of Definition~\ref{def:BG} for $(\textup{P}_{\textup{III}},D_7)$ if $q(x)$ solves $(\textup{P}_{\textup{III}},D_7)$ with the parameter~$\delta=0$. 
\end{proposition}

\begin{proof}
The system is obtained implementing the transformation~\hyperref[eq:tran1]{$(b)$} with~$(y_3,z_3)\mapsto (y,z)$ at the point $p_{\textup{III-}7}^{\textup{Ok}2}$ in~\eqref{eq:p3D7-1}. 

From $\mathbb{P}^2$, we identify the base points of the system~\eqref{eq:BG_P3_D72}, 
coinciding with the origins of the affine charts $(u_0,v_0)$ and $(U_0,V_0)$ respectively. From the base point $b_1$ we obtain two branches: $b_1^{(1)}$ as in equation~\eqref{eq:P3D6_casc_b1},  and $b_1^{(2)}$, that is different from~\eqref{eq:P3D6_casc_b2} and given in the following
\begin{equation}
    \label{eq:P3D7-2_casc_b1}
    \begin{aligned} 
    &b_1^{(2)}\colon(u_0\,,v_0) = (0\,,0) ~ \leftarrow ~ \big(u_1^{(2)},v_1^{(2)}\big) = (0\,,-q) ~ \leftarrow ~ \big(u_2^{(2)},v_2^{(2)}\big)  = \left( 0\,, 0 \right) ~ \leftarrow ~ \\[1ex]
    &~~~ \leftarrow ~ \big(U_3^{(2)},V_3^{(2)}\big) = \left(0\,,0\right) ~ \leftarrow ~ \big(u_4^{(2)},v_4^{(2)}\big) = \left(0\,,\frac{2q^{10}}{x\,\kappa_0}\right) ~ \leftarrow ~ \big(u_5^{(2)},v_5^{(2)}\big) = \left(0\,,-\frac{2q^{14}}{x\,\kappa_0}\right).
    \end{aligned} 
\end{equation}
The system in the chart $\big(u_6^{(2)},v_6^{(2)}\big)$ is regular. 
The last cascade emerging from the origin $(U_0,V_0)$  splits in two branches, as 
\begin{align}
\begin{aligned} 
    &b_2^{(\pm)} \colon (U_0\,,V_0) = (0\,,0) ~ \leftarrow ~ \big(U_1^{(\pm)},V_1^{(\pm)}\big) = (0\,,0) ~ \leftarrow ~ \\[1ex]
    &~ \leftarrow ~ \big(U_2^{(\pm)},V_2^{(\pm)}\big) = \left(\pm \,\dfrac{x }{2}\eta_{\infty} -\frac{x\, q'}{4 q}-\dfrac{1}{4 q}\,,0\right) \\[1ex]
    &~ \leftarrow ~ \big(U_3^{(\pm)},V_3^{(\pm)}\big) = \left(\pm \frac{1}{8 \eta_{\infty} q^2}\left(\frac{x (q')^2}{q}-c q''(c)-q'+4\kappa_{0}\right)\pm\frac{1}{2} x \eta_{\infty} q\,,0\right).
\end{aligned} 
\end{align}
The condition for the systems in the final charts $(U_4^{(\pm)},V_4^{\pm})$ to be regularised is 
    \begin{equation}
    -x q''' q^2-2 q^2 q''+4 x \eta_{\infty}^2 q^4 q'-8 \kappa_{0} q q'+3 q (q')^2-3 x (q')^3+4 x q q' q''+4 \eta_{\infty}^2 q^5 = 0\,, 
\end{equation}
verified by $q(x)$ solving {\hyperref[eq:P3_D7]{$(\text{P}_{\text{III}},D_7)$}} with $\delta=0$.
\end{proof}

Lastly, we consider the equation {\hyperref[eq:P3_D8]{$(\text{P}_{\text{III}},D_8)$}} as in \cite{Ohyama}, for which the Hamiltonian is not polynomial and is given by 
\begin{equation}\label{eq:H3D8}
  H^{\text{Ok}}_{\text{III-}8}\big( y_3,z_3;x \big)   = \frac{1}{x} \left(y_3^2 z_3^2+y_3 z_3-\frac{1}{2}\left( \frac{1}{y_3}+\frac{y_3}{x} \right)\right),
\end{equation}
with the parameters $\alpha, \beta$ being 
\begin{equation}\label{eq:coeff_D8}
    \alpha = 1\,, \qquad \beta = -1\,. 
\end{equation}
The equation satisfied by $y_3$ is the following 
\begin{equation}\label{eq:P3D8_from_syst}
    y_3'' = \frac{(y_3')^2}{y_3} - \frac{y_3}{x} +\frac{y_3^2}{x^2} -\frac{1}{x}\,. 
\end{equation}
\begin{proposition}
Applying the transformation \hyperref[eq:tran1]{$(b)$} of type I to the Hamiltonian system with Hamiltonian $H_{\textup{III-}8}^{\textup{Ok}}$~in \eqref{eq:H3D8} for the chart $(y_{3},z_{3})$ along the orbit $\mathcal{O}_{\text{III-8}}$,
\begin{equation}
    p_{\textup{III-}8}^{\textup{Ok}}\colon (y_3,z_3) = \left( q\,,\dfrac{x\,q'}{2 q^2}-\dfrac{1}{2 q}\right),
\end{equation}
gives rise to the system   
\begin{equation}\label{eq:BG_P3_D8}
   \hspace*{-2.8ex} \begin{cases}
	    \begin{aligned} 
        y' &=\dfrac{1}{(z+q)^2} \left(\frac{q'}{q} \left(\frac{z}{q} \left(2-3 y z^2\right)+\frac{z^2}{q^2}-10 y z^2+1\right)-yq' (4 q+11 z) \right.  \\[2ex]
        &~~-\frac{x (q')^2}{q^2} \left(\frac{z^2}{2 q^2}+\frac{z}{q}+\frac{1}{2}\right)+\frac{y q}{x} \left(2 q \left(1-9 y z^2\right)-2 y q^2 (q+5 z)+7 z \left(1-2 y z^2\right)\right)\\[2ex]
        &~~\left. -\frac{1}{x}\left(\frac{1}{2}-y z^2 \left(\frac{3 z}{q}-4 y z^2+8\right)\right)+\left(1-\frac{z}{x}\right) \left(\frac{z}{2 q^2}+\frac{1}{q}\right) \right),\\[2ex]  
        \end{aligned} \\[10ex]
        \begin{aligned}
            z' &= \dfrac{q'}{q}\left(\dfrac{z^2}{q}+2 z\right)+\frac{4 y z}{x} (q+z)^2
            -\frac{z}{x} \left(\frac{z}{q}+1\right) \,,
            \end{aligned}
	\end{cases}
\end{equation} 
which is Bureau-Guillot in the sense of Definition~\ref{def:BG} for $(\textup{P}_{\textup{III}},D_8)$ if $q(x)$ in turn  solves~$(\textup{P}_{\textup{III}},D_8)$.
\end{proposition}

\begin{proof}
We consider the transformation~\hyperref[eq:tran1]{$(b)$} with~$(y_3,z_3)\mapsto (y,z)$.  

The base points $b_1$, $b_2$ in $\mathbb{P}^2$ are in the origins of the affine charts $(u_0,v_0)$ and $(U_0,V_0)$ respectively. From $b_1$ we have two branches: $b_1^{(1)}$, same as in equation~\eqref{eq:P3D6_casc_b1}, and $b_1^{(2)}$ given in the following
\begin{equation}
    \begin{aligned}
      \begin{aligned} 
    &b_1^{(2)}\colon(u_0\,,v_0) = (0\,,0) ~ \leftarrow ~ \big(u_1^{(2)},v_1^{(2)}\big) = (0\,,-q) ~ \leftarrow ~ \big(U_2^{(2)},V_2^{(2)}\big)  = \left( 0\,, 0 \right) ~ \leftarrow ~ \\[1ex]
    &~~~ \leftarrow ~ \big(u_3^{(2)},v_3^{(2)}\big) = \left(0\,,0\right) ~ \leftarrow ~ \big(\widehat{U}_4^{(2)},\widehat{V}_4^{(2)}\big) = \left(\frac{x}{2q^2}\,, 0\right) ~ \leftarrow ~ \big({U}_5^{(2)},{V}_5^{(2)}\big) = \left( 0\,, 0 \right) , 
    \end{aligned}   
    \end{aligned}
\end{equation}
with the additional change of variables $\big(U_4^{(2)}, V_4^{(2)} \big) \mapsto \big((\widehat{U}_4^{(2)})^{-1}, \widehat{V}_4^{(2)} \big)$. The system in the final chart with coordinates $\big( U_6^{(2)} \,, V_6^{(2)} \big)$ is regular without generating any additional condition on the functions. 

From the base point $b_2$ a single cascade emerges, whose blow-ups are 
 \begin{align}
\begin{aligned} 
    &b_2^{(3)} \colon (U_0\,,V_0) = (0\,,0) ~ \leftarrow ~ \big(U_1^{(3)},V_1^{(3)}\big) = (0\,,0) ~ \leftarrow ~ \big(\widehat{u}_2^{(3)},\widehat{v}_2^{(3)}\big) = \left(0\,, \dfrac{2q^2}{q-xq'}\right)\\[.2ex]
    &~\leftarrow ~ \big(u_3^{(3)},v_3^{(3)}\big) = (0\,,0) ~\leftarrow ~  \big(U_4^{(3)},V_4^{(3)}\big) = \left( \frac{x}{2 q^2} \left(x q''+q'-\frac{x (q')^2}{q}+1\right) , 0 \right)  \\[.2ex]
    &~ \leftarrow ~ \big(U_5^{(3)},V_5^{(3)}\big) = \Bigg(  \frac{x}{q^2} \left(\frac{x q'}{q } \left(\frac{2 q'' }{q }+\frac{5 q'}{2}-\frac{3  x (q')^2}{q }+1\right)-\frac{1}{2} x q''' -2  x q''-q' -1\right),0 \Bigg), 
\end{aligned} 
\end{align}
where we used the intermediate change of variables $\big( u_2^{(2)}, v_2^{(2)} \big) \mapsto \big( (\widehat{u}_2^{(2)})^{-1}, \widehat{v}_2^{(2)} \big)$. 

The condition for the system in the chart $\big( U_6^{(2)}, V_6^{(2)} \big)$ to be regular is 
\begin{equation}
    \begin{aligned}
        &x^2 q^{(IV)} q^3 \left(q \left(q'+1\right)+ x \left(q q''-(q')^2\right)\right)+3  x \left(x^2 (q')^6-3  x q (q')^5+(3 q-2 x) q (q')^4\right)\\
        &-x^3 (q''')^2 q^4-2 q^3 \left(2 x^3 (q'')^3+q'\right) +2  x q''' q^3 \left(2 q \left(q'+1\right)+ x q' \left( x q''-2 q'-1\right)\right)
        \\
        &+q q''\left(2 q^2 \left((q-7 x) q'+q-x \right)+ x \left(q (13 x-11 q)-9  x q' \left( x q'-2 q\right)\right) (q')^2\right) \\
        &+ x q^2 (q'')^2 \left(3  x q' \left(3  x q'-2 q\right)-2 q (q+3 x)\right) +q^2 (q')^2\left((11 x-3 q) q'+(2 x-5 q)\right) =0\,,
    \end{aligned}
\end{equation}
satisfied if $q(x)$ solves {\hyperref[eq:P3_D8]{$(\text{P}_{\text{III}},D_8)$}} as in~\eqref{eq:P3D8_from_syst}. 
\end{proof}

Lastly, we consider the possibility of producing different systems using a generalisation of the transformation of~\hyperref[eq:tran1]{type I}, as we did in Section~\ref{sec:BGgen} for the case of the systems related to~$(\text{P}_{\text{I}})$ and~$(\text{P}_{\text{II}})$.

\begin{remark}
    It is possible to produce parametrised versions of the Bureau-Guillot systems obtained above. For instance, for the case {\hyperref[eq:P3_D6]{$(\text{P}_{\text{III}},D_6)$}}, we consider the parametrised transformation of type I of the chart in the variables $(y_3,z_3)$ for the Hamiltonian system with Hamiltonian~\eqref{eq:H3ok} along the orbit $\mathcal{O}_{\text{III-}6}$  in~\eqref{eq:point3Ok}, i.e. 
\begin{equation}
    \begin{cases}
        y_3= a_1\,z + q\,,      \\[1ex] 
        z_3 = a_2 \, yz + \dfrac{x\,\eta_{\infty}}{2}+\dfrac{x(q'-2  \eta_{0})}{4 q^2}+\dfrac{2 \kappa_{0}+1}{4 q}\,,
    \end{cases}
\end{equation}
with $a_1$, $a_2$ constants in $\mathbb{C}$. 
We get a new system also depending on the coefficients $a_1, a_2 \in \mathbb{C}$. The condition for the latter to be regularised is verified when $q(x)$ solves {\hyperref[eq:P3_D6]{$(\text{P}_{\text{III}},D_6)$}}.
\end{remark}
We find similar results for both versions of {\hyperref[eq:P3_D7]{$(\text{P}_{\text{III}},D_7)$}} and {\hyperref[eq:P3_D8]{$(\text{P}_{\text{III}},D_8)$}}, with Hamiltonians $H^{\text{Ok}1,2}_{\text{III-}7}$ in~\eqref{eq:H3D7}, \eqref{eq:H3D72} and $H^{\text{Ok}}_{\text{III-}8}$ in~\eqref{eq:H3D8} respectively. The same generalisation works for systems resulting from taking into account the rational Hamiltonians, i.e.\ $H^{\text{rat}}_{\text{III-}6}$ in~\eqref{eq:HamP3_rational} for systems related to {\hyperref[eq:P3_D6]{$(\text{P}_{\text{III}},D_6)$}} and the corresponding Hamiltonians for the two versions of~{\hyperref[eq:P3_D7]{$(\text{P}_{\text{III}},D_7)$}} (in~\eqref{eq:HamP3_rational} either $\gamma=0$ or $\delta=0$) and {\hyperref[eq:P3_D8]{$(\text{P}_{\text{III}},D_8)$}} (in~\eqref{eq:HamP3_rational} taking $\gamma=\delta=0$).  

\subsection{\texorpdfstring{Systems for $(\text{P}_{\text{IV}})$, $(\text{P}_{\text{V}})$ and $(\text{P}_{\text{VI}})$}{p4p5p6}}
The Painlev\'e IV equation has the form 
\begin{equation}\label{eq:P4}
    (\text{P}_{\text{IV}})\colon ~~ y''= \dfrac{(y')^2}{2\,y} + \dfrac{3}{2} \,y^3 + 4x\,y^2 +2(x^2-\alpha)y + \dfrac{\beta}{y}\,,\qquad   y \in \mathbb{C}\,,
\end{equation}
with  $\alpha, \beta, \gamma, \delta \in \mathbb{C}$ constant parameters. 
We consider the Okamoto Hamiltonians for $(\text{P}_{\text{IV}})$ as in~\cite{okamoto3,Ohyama_2006} in the variables $(y_4,z_4)$
\begin{equation}\label{eq:H4}
    H_{\text{IV}}^{\text{Ok}}(y_4\,,z_4\,;x)= 2\,y_4z_4^2 - \left(y_4^2 + 2x_4 y_4 + \kappa_0 \right)z_4 + \kappa_\infty y_4 \,,
\end{equation}
with the following relations between the parameters $\kappa_0, \kappa_\infty$ and $\alpha, \beta$:
\begin{equation}
    \alpha = 1 - \dfrac{\kappa_0}{2} +2\kappa_\infty \,, \qquad \beta = - \dfrac{\kappa_0^2}{2} \,. 
\end{equation}

\begin{proposition}\label{Prop3.5}
Applying the transformation \hyperref[eq:tran1]{$(b)$} of type I to the Hamiltonian system with Hamiltonian~$H_{\textup{IV}}^{\textup{Ok}}$~in \eqref{eq:H4} for the chart $(y_{4},z_{4})$ along the orbit $\mathcal{O}_{\text{IV}}$,
\begin{equation}\label{eq:point_H4}
    p_{\textup{IV}}^{\textup{Ok}}\colon (y_4,z_4) = \left( q\,,\frac{x }{2}+\frac{q}{4}+\frac{\kappa_0+q'}{q}\right),
\end{equation}
gives rise to a system that is Bureau-Guillot in the sense of Definition~\ref{def:BG} for $(\textup{P}_{\textup{IV}})$ if $q(x)$ is in turn the~$\textup{P}_{\textup{IV}}$ transcendent. 
\end{proposition}
\begin{proof}
We consider the blow-up transformation~\hyperref[eq:tran1]{$(b)$} for the chart $(y_4,z_4)$ at the point~\eqref{eq:point_H4}, i.e.\ 
\begin{equation}
    \begin{cases}
        y_4 = z +  q\,, \\[2ex]
        z_4 = yz + \dfrac{x }{2}+\dfrac{q}{4}+\dfrac{\kappa_0+q'}{q},
    \end{cases}
\end{equation}
obtaining the system of first order ODEs in $(y,z)$: 
{
\begin{equation}
    \begin{cases}
        \begin{aligned} 
        y'=\dfrac{q}{2}\left(1+4 y-8 y^2\right)-\dfrac{(4 y-1) \left(\kappa_{0}+q'\right)}{2 q} +3 y z(1-2 y) +x
        \end{aligned} \,, \\[2ex]
        z'=\dfrac{z}{q} \big(\kappa_{0}+q'+q (q+z)(4 y-1) \big)\,. 
    \end{cases}
\end{equation}
Analysing the system in $\mathbb{P}^2$, we identify the base points $b_1\colon(u_0,v_0) = (0,0)$ and $b_2\colon(U_0,V_0) = (0,0)$, both giving origin to two cascades of blow-ups. In particular, we have the following: 
\begin{align}
\label{eq:P4_casc_b1}
    b_1^{(1)}&\colon(u_0\,,v_0) = (0\,,0) ~ \leftarrow ~ \big(u_1^{(1)},v_1^{(1)}\big) = (0\,,0)\,, \\[1ex]
    \label{eq:P4_casc_b12}
    \begin{split} 
    b_1^{(2)}&\colon(u_0\,,v_0) = (0\,,0) ~ \leftarrow ~ \big(u_1^{(2)},v_1^{(2)}\big) = (0\,,-q\,) ~ \leftarrow ~ \big(u_2^{(2)},v_2^{(2)}\big)  = \left( 0\,, -\dfrac{\kappa_0}{2q} \right), 
    \end{split} 
\end{align}
with regular final chart systems in $\big( u^{(1)}_2\,, v^{(1)}_2 \big)$ and $\big( u^{(2)}_3\,, v^{(2)}_3 \big)$ respectively. From the remaining base point in $\mathbb{P}^2$ the cascades are 
\begin{align}
\label{eq:P4_casc_b21}
\begin{split} 
    b_2^{(3)}&\colon(U_0\,,V_0) = (0\,,0) ~ \leftarrow ~ \big(U_1^{(3)},V_1^{(3)}\big) = (0\,,0)\, ~ \leftarrow ~ \big(U_2^{(3)},V_2^{(3)}\big)  = \left( -\frac{q^2+2 q x+\kappa_{0}+q'}{4 q}\,, 0 \right) ~ \leftarrow ~  \\[1ex]
    & ~~ \leftarrow ~ \big( U^{(3)}_3\,, V^{(3)}_3 \big) = \left(  -\dfrac{ (q-\kappa_{0})^2 -q^2\left(4 c^2+8 c q+3 q^2-3\right)+2 q q''-(q')^2}{8q}\,, 0 \right),
    \end{split} 
\end{align}
\begin{align}
    \label{eq:P4_casc_b22}
    \begin{split} 
    b_2^{(4)}&\colon(U_0\,,V_0) =  (0\,,0) ~ \leftarrow ~ \big(U_1^{(4)},V_1^{(4)}\big) = \left(\dfrac{1}{2}\,,0\right)\, ~ \leftarrow ~ \big(U_2^{(4)},V_2^{(4)}\big)  = \left( \frac{q^2+2 q x-(\kappa_{0}+q')}{4 q}\,, 0 \right) ~ \leftarrow ~  \\[1ex]
    & ~~ \leftarrow ~ \big( U^{(4)}_3\,, V^{(4)}_3 \big) = \left(  \dfrac{ (q+\kappa_{0})^2 -q^2\left(4 c^2+8 c q+3 q^2+5\right)+2 q q''-(q')^2}{8q}\,, 0 \right).
    \end{split} 
\end{align}
At the end of both cascades, the condition for the systems expressed in the variables $\big( U^{(3)}_4\,, V^{(3)}_4 \big)$ and $\big( U^{(4)}_4\,, V^{(4)}_4 \big)$ respectively, is expressed in terms of the coefficient function $q(x)$ and its derivatives as 
\begin{equation}
    q''' q^2-2 q q' q''+(q')^3 -\kappa_{0}^2 q'- q^4 (3q'+4)-4 x q^3( q'+1)= 0 \,,
\end{equation}
which is satisfied if $q(x)$ solves \hyperref[eq:P4]{$(\text{P}_{\text{IV}})$}. 

}

\end{proof}

The Painlev\'e V equation has the form 
\begin{equation}\label{eq:P5}
    (\text{P}_{\text{V}})\colon ~~ y''= \left(\frac{1}{2y}+\frac{1}{y-1}
\right)\left(y'\right)^{2}-\frac{y'}{x}+\frac{(y-1)^{2}}{x^{2}}\left(\alpha y+\frac{\beta}{y%
}\right)+\frac{\gamma y}{x}+\frac{\delta y(y+1)}{y-1}\,,~~ y \in \mathbb{C}\,,
\end{equation}
with $\alpha, \beta, \gamma, \delta \in \mathbb{C}$ constant parameters. 
We consider the Okamoto Hamiltonian for $(\text{P}_{\text{V}})$ as in~\cite{okamoto2,Ohyama_2006} in the variables $(y_5,z_5)$
\begin{equation}\label{eq:H5}
    H_{\text{V}}^{\text{Ok}}(y_5\,,z_5\,;x)= \frac{1}{x}\Big(y_5(y_5-1)^2 z_5^2-z_5 \big(\kappa_{0} (y_5-1)^2+\kappa_{t}\, y_5 (y_5-1)-\eta\,  x y_5\big)+\kappa  (y_5-1)\Big)\,,
\end{equation}
with the following relations between the parameters $\kappa_0, \kappa_t, \kappa, \eta$ and $\alpha, \beta, \gamma, \delta$:
\begin{equation}
\alpha=\frac{1}{2}(\kappa_0+\kappa_t)^2-2\kappa \,, \quad \beta=-\frac{1}{2}\kappa_{0}^2 \,, \quad  \gamma=-\eta(\kappa_t+1)    \,, \quad\delta=-\frac{1}{2}\eta^2 \,, 
\end{equation}
 with 
$$\kappa=\frac{(\kappa_0+\kappa_t)^2}{4}-\frac{\kappa_{\infty}^2}{4}, \quad \eta\neq 0\,.$$
\begin{proposition}\label{prop:P5}
Applying the transformation \hyperref[eq:tran1]{$(b)$} of type I to the Hamiltonian system with Hamiltonian~$H_{\textup{V}}^{\textup{Ok}}$~in \eqref{eq:H5} for the chart $(y_{5},z_{5})$ along the orbit $\mathcal{O}_{\text{V}}$,
\begin{equation}\label{eq:point_H5}
    p_{\textup{V}}^{\textup{Ok}}\colon \big(y_5,z_5\big) = \left( q\,,\frac{x q'-q (x \eta +2 \kappa_{0}+\kappa_{t})+q^2 (\kappa_{0}+\kappa_{t})+\kappa_{0}}{2 q(q-1)^2 }\right),
\end{equation}
gives rise to a system that is Bureau-Guillot in the sense of Definition~\ref{def:BG} for $(\textup{P}_{\textup{V}})$ if $q$ is in turn a~$\textup{P}_{\textup{V}}$ transcendent. 
\end{proposition}
\begin{proof}
    We blow-up the chart $(y_5,z_5)$ at the point~\eqref{eq:point_H5} applying the transformation~\hyperref[eq:tran1]{$(b)$} in the coordinates~$(y_5,z_5)\mapsto (y,z)$
\begin{equation}\label{eq:transf_P5}
    \begin{cases}
        y_5 = z +  q\,, \\[2ex]
        z_5  = yz + \dfrac{x q'-q (x \eta +2 \kappa_{0}+\kappa_{t})+q^2 (\kappa_{0}+\kappa_{t})+\kappa_{0}}{2 q(q-1)^2 }\,. 
    \end{cases}
\end{equation}
To complete the proof, we determine the system in the variables $(y,z)$, and find that the condition for this to be regularised is that $q(x)$ solves $(\text{P}_{\text{V}})$. The system and the regularising condition are reported in Appendix~\ref{app:P5}.

\end{proof}

The Painlev\'e VI equation has the form 
\begin{equation}
\begin{aligned}
    (\text{P}_{\text{VI}})\colon ~~ y''&= \frac{1}{2}\left(\frac{1}{y}+\frac%
{1}{y-1}+\frac{1}{y-x}\right)\left(y'\right)^{2}-%
\left(\frac{1}{x}+\frac{1}{x-1}+\frac{1}{y-x}\right)y'\\[1ex]
&~~+\frac{y(y-1)(y-x)}{x^{2}(x-1)^{2}}\left(\alpha+\frac{\beta x}{y^{2}}+%
\frac{\gamma(x-1)}{(y-1)^{2}}+\frac{\delta x(x-1)}{(y-x)^{2}}\right),~~   y \in \mathbb{C}\,,
\end{aligned} 
\end{equation}
with $\alpha$, $\beta$, $\gamma$, $\delta \in \mathbb{C}$ constant parameters.
We consider the Okamoto Hamiltonians for $(\text{P}_{\text{VI}})$ as in~\cite{okamoto1,Ohyama_2006} in the variables $(y_6,z_6)$
\begin{equation}\label{eq:H6}
\begin{aligned} 
    H_{\text{VI}}^{\text{Ok}}(y_6\,,z_6\,;x)&= \frac{1}{x(x - 1)} \Big(y_6(y_6 - 1)(y_6 - x)z_6^2\\
    &~~-z_6\big(\kappa_0(y_6 - 1)(y_6 - x) + \kappa_1 y_6(y_6- x)+(\kappa_t-1)y_6(y_6-1)\big) + \kappa(y_6 - x)\Big)\,,
\end{aligned} 
\end{equation}
with the following relations between the parameters $\kappa_0,\kappa_1, \kappa_t, \kappa$ and $\alpha, \beta, \gamma, \delta$:
\begin{equation}
 \alpha=\frac{1}{2} (\kappa_t +\kappa_{0}+\kappa_{1}-1)^2-2\kappa \,, \quad \beta=-\frac{1}{2}\kappa_{0}^2 \,, \quad \gamma=\frac{1}{2}\kappa_1^2  \,, \quad\delta=\frac{1}{2}(1-\kappa_t^2) \, 
\end{equation}
with $$\kappa =\frac{1}{4} \left(\kappa _0+\kappa _1+\kappa _t-1\right){}^2-\frac{\kappa _{\infty }^2}{4}.$$
\begin{proposition}\label{prop:P6}
    Applying the transformation \hyperref[eq:tran1]{$(b)$} of type I to the Hamiltonian system with Hamiltonian $H_{\textup{VI}}^{\textup{Ok}}$~in \eqref{eq:H6} for the chart $(y_{6},z_{6})$ along the orbit $\mathcal{O}_{\text{VI}}$,
\begin{equation}\label{eq:point_H6}
  \!\!  p_{\textup{VI}}^{\textup{Ok}}\colon \big(y_6,z_6\big) = \left( q\,, \frac{q^2 (1-\kappa_{0}-\kappa_{1}-\kappa_{t})-q (1-\kappa_{0}(x+1)-\kappa_{t}-\kappa_{1} x)-x (q'+\kappa_{0})}{2q(q-1)(q-x)}\right),
\end{equation}
gives rise to a system that is Bureau-Guillot in the sense of Definition~\ref{def:BG} for $(\textup{P}_{\textup{VI}})$ if $q(x)$ is in turn a~$\textup{P}_{\textup{VI}}$ transcendent. 
\end{proposition}
\begin{proof}
     From the chart $(y_6,z_6)$ we get 
\begin{equation}\label{eq:transf_P6}
    \begin{cases}
        y_6 = z +  q\,, \\[2ex]
        z_6  = yz + \dfrac{x q'-q (x \eta +2 \kappa_{0}+\kappa_{t})+q^2 (\kappa_{0}+\kappa_{t})+\kappa_{0}}{2 q(q-1)^2 }\,, 
    \end{cases}
\end{equation}
and we generate the system of first order ODEs in $(y,z)$. The details of the system and the regularising condition are reported in Appendix~\ref{app:P6}.
\end{proof}

We conclude this Section with the following observation. 

\begin{remark}
    We can construct the generalised versions of the Bureau-Guillot systems with Painlev\'e transcendents $\text{P}_{\text{IV}}$, $\text{P}_{\text{V}}$ and $\text{P}_{\text{VI}}$ in the coefficients that we list in the section, by performing the parametrised version of the transformation of type I. For instance, from the Hamiltonian $H^{\text{Ok}}_{\text{IV}}$ in \eqref{eq:H4} in the chart~$(y_4,z_5)$ along the orbit with points $p^{\text{Ok}}_{\text{IV}}$ in~\eqref{eq:point_H4},  we consider the change of variables 
    \begin{equation}
        \begin{cases}
            y_4 = a_1\, z + q \,,\\[1ex]
            z_4 = a_2\,yz + \dfrac{x \eta_{\infty}}{2}+\dfrac{x q'-2 x \eta_{0}}{4 q^2}+\dfrac{2 \kappa_{0}+1}{4 q}\,,
        \end{cases}
    \end{equation}
    with $a_1,a_2 \in \mathbb{C}$ two constants. The system we produce is the parametrised version of a Bureau-Guillot system in the sense of Definition~\ref{def:BG}. 
\end{remark}

\section{Generalisations and mixed cases}
\label{sec:gen_mixedcases} 

This Section extends the method described above both in terms of the initial system, and in terms of the transformation of variables. First, we generalise the type of the initial system, considering rational Hamiltonian systems and non-Hamiltonian cases. Second, we introduce the transformation of type II: a simpler, shift-like or linear  (rather than blow-up-like) transformation of the affine variables, enabling the construction of more general systems. This includes mixed cases and those whose regularising conditions yield quasi-Painlev\'e equations.  

\subsection{Examples with different initial systems}
The method we introduced to generate systems with Painlev\'e transcendents in the coefficients, is not constrained to the cases studied in the previous sections. Indeed, one can successfully produce analogous systems by following the proposed construction for the rational Hamiltonians for Painlev\'e equations. We give an example of this case, for {\hyperref[eq:P3_D6]{$(\text{P}_{\text{III}},D_6)$}}. 

\begin{example}
    We build a system with the  $\text{P}_{\text{III-}6}$ transcendent in the coefficients, starting from the rational Hamiltonian: 
\begin{equation}\label{eq:HamP3_rational}
H_{\text{III-}6}^{\text{rat}}\big(y_3,z_3;x\big)=\dfrac{1}{2}\dfrac{y_3^2\, z_3^2}{x} -\dfrac{\alpha \, y_3^2-\beta}{y_3}-\dfrac{\gamma}{2} \,  x_3 y_3^2+\dfrac{\delta}{2}  \,\dfrac{x}{y_3^2}\,, 
\end{equation} 
for which the system in $(y_3,z_3)$ is 
\begin{equation}
\begin{cases}
    y_3'=\dfrac{y_3^2 z_3}{x}, \\[1ex] 
    z_3'=\alpha +\dfrac{\beta }{y_3^2}+  \gamma\,x\,  y_3+\delta\,\dfrac{x}{y_3^3}-\dfrac{y_3 z_3^2}{x} .
\end{cases}  
\end{equation}
The variable $y_3$ satisfies {\hyperref[eq:P3_D6]{$(\text{P}_{\text{III}},D_6)$}}. With $q(x)$ being a  solution to {\hyperref[eq:P3_D6]{$(\text{P}_{\text{III}},D_6)$}} as well, we apply the transformation~\hyperref[eq:tran1]{$(b)$} of type I along the orbit $\mathcal{O}_{\text{III-}6}$,
\begin{equation}
    p_{\text{III}}^{\text{rat}}\colon (y_3, z_3) = \left( q\,, \dfrac{x\,q'}{q^2} \right),
\end{equation}
obtaining the following change of variables
\begin{equation}
    \begin{cases}
        y_3 = z + q\,,\\[1ex]
        z_3 = yz + \dfrac{x\,q'}{q^2}\,.
    \end{cases}
\end{equation}
The new system in $(y,z)$ can be regularised if $q(x)$ solves {\hyperref[eq:P3_D6]{$(\text{P}_{\text{III}},D_6)$}}. 
\end{example}

The Hamiltonian property of the system does not seem to be necessary for our construction. Instead, what is relevant is the possibility of 
deriving a second order differential equation for one of the variables, as reported in Section~\ref{sec:generalconstruction}. In the next Example we treat a system derived in~\cite{FilipukChen} in the context of Meixner orthogonal polynomials, non-Hamiltonian with respect to the canonical 2-form, and whose second order equation is related to $(\text{P}_{\text{V}})$ via a M\"obius transformation.   

\begin{example}
We consider the system~\cite{FilipukChen} in the variables $(y_5,z_5)$
\begin{equation}\label{eq:syst_meixner}
    \begin{cases}
        y_5'=\dfrac{1}{x^2}\left( (y_5+1)(n x+x^2y_5) +2 \gamma  y_5(x+z_5)+y_5(2z_5-x -\beta  x )\right),\\[2ex]
        z_5'=\dfrac{1}{x^2 (y_5+1)}\left(x^2(-\beta  +\gamma  -z_5(y_5^2 -2  y_5 - 1))+x z_5(\beta+1)+ (1-\gamma )(z_5^2+2xz_5)\right),
    \end{cases}
\end{equation}
with $\beta$, $\gamma \in \mathbb{C}$, and $n\in \mathbb{N}$. The orbit of a solution $q$ to the system is given by the point with coordinates 
\begin{equation}\label{eq:meix_point}
    p_{\text{Meix}}\colon (y_5, z_5) = \left( q\,,\frac{x \left(q (2 \gamma -\beta  +n+x-1)+n-x q'+x q^2\right)}{2 (\gamma -1) q}\right) \,. 
\end{equation}
Applying the transformation \hyperref[eq:tran1]{$(b)$} of type I to the system~\eqref{eq:syst_meixner} at the point~\eqref{eq:meix_point}, i.e.
\begin{equation}
    \begin{cases}
        y_5 = z + q \,, \\[1ex]
        z_5 = yz + \dfrac{x \left(q (2 \gamma -\beta  +n+x-1)+n-x q'+x q^2\right)}{2 (\gamma -1) q}\,,
    \end{cases}
\end{equation}
we obtain the following system in the variables $(y,z)$: 
\begin{equation}
    \begin{cases}
    \begin{aligned}
        y'&=\dfrac{1}{4 (\gamma -1) x^2 q^2 (q+1) (q+z+1)}\left(q^2 \left(x^2 \left((\beta -1)^2-n^2-10 n x-3 x^2+2 x (\beta -2 \gamma +1)\right) \right.\right. \\[1ex]
        &~~ \left. \left. +2 x^3 \left(z (\beta -2 (\gamma + n)-x+1)+x (z+3) q'\right)+4 (\gamma -1) x y \left(z (n-2 x z-3 x+1) +4 n\right.\right. \right.\\[1ex]
        &~~ \left. \left.\left. -x (z+3) q'+1\right)+4 (\gamma -1)^2 y^2 z (2 z+3)\right)+q^3 \left(2 x^3 \left(2 \beta -4 \gamma +z (\beta -2 (\gamma+x) -n+1)\right.\right. \right.\\[1ex]
        &~~ \left. \left.\left.-4 n+x q'-4 x+2\right)-4 (\gamma -1) x y \left(-2 (n-x+1)+x q'+z (2 x z+7 x-1)\right)\right.\right. \\[1ex]
        &~~  \left.\left.+4 (\gamma -1)^2 y^2 (z (2 z+7)+2)\right)-x^2 \left(n-x q'\right)^2-2 x q \left(n-x q'\right) (x (n+x z+2 x)\right. \\[1ex]
        &~~  \left.-2 (\gamma -1) y (z+2))-q^4 \left(\left(x^3 (-2 \beta +4 \gamma +2 n+2 x z+7 x-2)\right)-16 (\gamma -1)^2 y^2 (z+1)\right.\right. \\[1ex]
        &~~  \left.\left.+4 (\gamma -1) x y (4 x z+4 x-1)\right)-2 q^5 \left(x^4+4 (\gamma -1) y \left(x^2-\gamma  y+y\right)\right)\right),
        \end{aligned} \\
        \\
        z'=\dfrac{z \left(x^2 q'-n x+q (q+z) \left(x^2-2 (\gamma -1) y\right)\right)}{x^2 q}\,.
    \end{cases}
\end{equation}
The regularising condition at the end of the cascades is verified if $q(x)$ solves the second order equation for the system~\eqref{eq:syst_meixner}, i.e. 
\begin{equation}
\begin{split} 
    q'' &= \frac{1}{2 x^2 q (q+1)}\left(q^2 \left(\beta ( \beta-2(x+1)) -n(n-2 x)-2 x (q'-2\gamma)+x^2+1\right)+x^2 (q'^2+2  q^5) \right.  \\[1ex]
    &\qquad \left.  -n^2+x q^4 (4 \gamma-2 \beta  +2 n+5 x)+4 x q^3 (2 \gamma -\beta +n+x)-2 q \left(n^2-x^2 q'^2+x q'\right) \right). 
\end{split} 
\end{equation}

\end{example}

\subsection{Examples of mixed cases }\label{sec:mixedcases}
 In this section we address the question whether the construction proposed above can be replicated to produce mixed systems, i.e.\ systems   related to a certain Painlev\'e equation $(\text{P}_{\text{J}})$ with $\text{J}=\text{I}, \dots , \text{VI}$, whose coefficient function $q(x)$ solves a different  Painlev\'e equation~$(\text{P}_{\text{K}})$ with~$\text{K}\neq \text{J}$. 

As a first attempt, we apply the transformation of type I listed in~\eqref{eq:tran1} to the chart $(y_{\text{J}},z_{\text{J}})$ associated with~$(\text{P}_{\text{J}})$ along the orbit $\mathcal{O}_{\text{K}}$ for~$(\text{P}_{\text{K}})$. With this construction, starting from a polynomial system for~$(\text{P}_{\text{J}})$, we obtain a rational system whose regularisation is not ensured. In particular, we observe that one or more cascades of blow-ups are of finite length with regularising conditions such that $q(x)$ solves~$(\text{P}_{\text{K}})$. However, one or more cascades that cannot be regularised emerge, of potentially infinite length. Therefore, we propose an alternative change of variables.

We notice that the blow-up-like birational transformations reported in~\eqref{eq:tran1} to generate Bureau-Guillot systems in the sense of Definition~\ref{def:BG} are not the simplest possible transformations, these being instead shift-like transformations. In particular, for a system in the variables $(y_\text{J},z_\text{J})$ associated with~$(\textup{P}_\textup{J})$, and considering $p_{\textup{J}}\colon (q,h_\text{J}(q,q';x))$ a generic point of the orbit $\mathcal{O}_\text{J}$, these are the changes of variables $(y_\text{J},z_\text{J})\mapsto (y,z)$
\begin{equation}\label{eq:type2_1}
    (a)~~ \begin{cases}
        y_\text{J}= y + q \,, \\[1ex]
        z_\text{J}= z + h_\text{J}(q,q';x)\,, \\[1ex]
    \end{cases} \qquad (b)~~ \begin{cases}
        y_\text{J}= z + q \,, \\[1ex]
        z_\text{J}= y +h_\text{J}(q,q';x)\,. \\[1ex]
    \end{cases}
\end{equation}
The resulting system in $(y,z)$ is still related to $(\textup{P}_\textup{J})$ and one finds $q$ being in turn solution of $(\textup{P}_\textup{J})$ as a regularising condition. 

This type of transformations can be used to construct a working analogue of the Bureau-Guillot systems with the mixed case. In order to obtain it, we just transform the system associated with $(\textup{P}_\textup{J})$ using a change of variables built on the orbit $\mathcal{O}_\text{K}$ for $(\textup{P}_\textup{K})$ with $\textup{K} \neq \textup{J}$:  
\begin{equation}\label{eq:type2_2}
    (c)~~ \begin{cases}
        y_\text{J}= y + q \,, \\[1ex]
        z_\text{J}= z + h_\text{K}(q,q';x)\,, \\[1ex]
    \end{cases} \qquad (d)~~ \begin{cases}
        y_\text{J}= z + q \,, \\[1ex]
        z_\text{J}= y +h_\text{K}(q,q';x)\,. \\[1ex]
    \end{cases}
\end{equation}
  The system of ODEs in $(y,z)$ is associated with $(\textup{P}_\textup{J})$ with the coefficient function $q(x)$ solving $(\textup{P}_\textup{K})$ emerging as the regularising condition.  This simpler construction allows us to generate systems where the coefficient functions are more general, e.g.\ they can be solutions of quasi-Painlev\'e equations. We call the transformations in~\eqref{eq:type2_1} and~\eqref{eq:type2_2} of type II.

In the following we construct some examples for the mixed cases and one for the quasi-Painlev\'e extension.   

\begin{example}\label{ex:P1withP2}
    Let us consider the simplest system for $(\text{P}_{\text{I}})$ in the variables $(y_1,z_1)$ that we already analysed in Section~\ref{sec:BGgen} with the labelling of system V in equation~\eqref{eq:system_5}: 
\begin{equation}\label{eq:system_P1}
\begin{cases}
    y_1' = z_1,\\[1ex]
z_1' = 6\,y_1^2 + x.
\end{cases} 
\end{equation}
Let us now consider the simplest system for $(\text{P}_{\text{II}})$ in the variables $(y_2,z_2)$, identified as system IX.B(3) in equation~\eqref{eq:syst_93}
\begin{equation}\label{eq:syst_93_general}
\begin{cases} 
y_{2}' =  z_{2}-y_{2}^2  - \dfrac{x}{2}, \\[3ex]
z_{2}' = 2\,y_{2}\,z_{2} + \alpha+\dfrac{1}{2}.
\end{cases}
\end{equation}
We consider the transformation \hyperref[eq:type2_2]{$(d)$} of type II by including the point of the orbit for $(\text{P}_{\text{II}})$:
\begin{equation}
    p_{\text{II}} \colon (y_1,z_1) = \left( q\,, q^2 + q' +\frac{x}{2} \right), 
\end{equation}
with $q \equiv y_2$ in~\eqref{eq:syst_93_general} for $(y_1,z_1)\mapsto(y,z)$. The change of coordinates to consider is then 
\begin{equation}
    \begin{cases}
        y_1 = z + q, \\[2ex]
        z_1 = y + q^2 + q' +\dfrac{x}{2},
    \end{cases}
\end{equation}
generating the following system in the new variables $(y,z)$:
\begin{equation}\label{eq:system_ratio}
    \begin{cases}
    y'=q^2+z+\dfrac{x}{2},\\[1ex]
    z'= 6 y^2+12 q \, y- \alpha -2 q q'-2 q^3+6 q^2- x( q-1)-\dfrac{1}{2}.
    \end{cases}
\end{equation}
{In $\mathbb{P}^2$ we identify one base point, from which the following cascade originates:
\begin{align}
    \begin{split} 
    b_1&\colon(U_0,V_0)=(0\,,0) ~ \leftarrow ~ (u_1,v_1)=(0\,,0) ~ \leftarrow ~ (u_2,v_2)=(0\,,0) ~ \leftarrow ~ \\[1ex] 
    &~ \leftarrow ~ (u_3,v_3)=(0\,,4) ~ \leftarrow ~(u_4,v_4)=(0\,,0)~ \leftarrow ~ (u_5,v_5)=(0\,,48q) ~ \leftarrow ~ \\[1ex]
    &~ \leftarrow ~ (u_6,v_6)=(0\,,-16(2(q^2+q')+x)) ~ \leftarrow ~ \\[1ex]
    &~ \leftarrow ~ (u_7,v_7)=(0\,,-32(2q^3 -24q^2+xq-q''-x+\alpha)) ~ \leftarrow ~ \\[1ex]
    &~ \leftarrow ~ (u_8,v_8)=(-64 \left(q'''-6 q^2 q'+18 q q'-c q'+18 q^3+9 x q-q+1\right)\,,0).
    \end{split} 
\end{align}
In the final chart $(u_9,u_9)$ at the end of the cascade the resulting system is regular if the following condition is satisfied 
\begin{equation}
    q^{(IV)}-6 q^2 q''-x q''-12 q q'^2-2 q' = 0\,.
\end{equation}
This condition can be integrated twice in $x$, and by fixing one of the two integration constants to zero yields $(\text{P}_{\text{II}})$. 

}

\end{example}

\begin{example}
    Let us consider the system for $(\text{P}_{\text{IV}})$ given by the Okamoto Hamiltonian $H_{\text{IV}}^{\text{Ok}}$ in~\eqref{eq:H4} in the variables $(y_4,z_4)$
    \begin{equation}\label{eq:P4_system}
        \begin{cases}
        y_4'=4\, y_4 z_4-y_4^2 -2 x\, y_4 -\kappa_{0}\,,\\[1ex]
        z_4'=-2 z_4^2+2z_4 (x+y_4)-\kappa_{\infty}\,. 
        \end{cases}
    \end{equation}
    We rely on the system~\eqref{eq:system_P1} for $(\text{P}_{\text{I}})$ in Example~\ref{ex:P1withP2} in the variables $(y_1,z_1)$ to identify the orbit $\mathcal{O}_\text{I}$, i.e.\ for generic point 
    \begin{equation}
        p_{\text{I}}\colon (y_4,z_4) = \left( q\,, q'\right) \,, 
    \end{equation}
    with $q\equiv y_1$ in~\eqref{eq:system_P1}. We perform the transformation
    \begin{equation}
        \begin{cases}
            y_4 = y + q,\\[1ex]
            z_4 = z + q'.
        \end{cases}
    \end{equation}
    The system in the coordinates $\big(y,z\big)$ is 
    \begin{equation}
        \begin{cases}
            \begin{aligned}
               y' &= (q+y) \left(4 q'-q-2 x-y+4 z\right)-q'  -\kappa_{0},
            \end{aligned}\\[2ex]
            z' = 2 \left(q'+z\right) \left(-q'+q+x+y-z\right)-6 q^2-x-\kappa_{\infty} \,,     
            \end{cases}
    \end{equation}
and it is related to $(\text{P}_{\text{IV}})$ if the coefficient function $q(x)$ is the transcendent $\text{P}_{\text{I}}$. The latter arises as a regularising condition in two of the three cascades. In particular, considering the compactified space~$\mathbb{P}^2$ we identify the base points $b_1\colon (u_0,v_0)=(0,0)$, $b_2\colon (U_0,V_0)=(0,0)$ and $b_3\colon (U_0,V_0)=(2,0)$, the latter visible also in the chart $(u_0,v_0)$. The cascades of blow-ups are the following: 
\begin{align}
\begin{split} 
    b_1&\colon (u_0,v_0)=(0\,,0) ~ \leftarrow ~ (u_1^{(1)},v_1^{(1)})=(0\,,-q' ~ \leftarrow ~ (u_2^{(1)},v_2^{(1)})=(0\,,6q^2-q''+x+\kappa_{\infty}), 
 \end{split} \\[1ex]
    \begin{split}
    b_2&\colon (U_0,V_0)=(0\,,0) ~ \leftarrow ~ (U_1^{(2)},V_1^{(2)})=(-q\,,0) ~ \leftarrow ~ (U_2^{(2)},V_2^{(2)})=\left(\dfrac{\kappa_0}{2}\,,0 \right), 
    \end{split} 
    \\[1ex]
    \begin{split}
    b_3&\colon (U_0,V_0)=(2\,,0) ~ \leftarrow ~ (U_1^{(3)},V_1^{(3)})=(2q'-2q-2x\,,0) ~ \leftarrow ~ \\[1ex]
    &~ \leftarrow ~ (U_2^{(3)},V_2^{(3)})=\left(6q^2 -q'' -\dfrac{\kappa_0}{2}+\kappa_{\infty}+x-1\,,0 \right). 
    \end{split} 
\end{align}
The final system in the coordinates $(U_3^{(2)},V_3^{(2)})$ is regular, while for regularising the system in both final systems $(u_3^{(1)},v_3^{(1)})$ and $(u_3^{(3)},v_3^{(3)})$ the condition in terms of the coefficient function $q(x)$ is 
\begin{equation}
    q'''-12qq'-1 = 0\,,
\end{equation}
which can be easily integrated once in $x$, giving $(\text{P}_{\text{I}})$. 
    
\end{example}

\begin{example}\label{ex:p2p5}
    Let us consider the system for $(\text{P}_{\text{II}})$ as in~\eqref{eq:syst_93_general} in the variables $(y_2,z_2)$. 
    We identify the orbit $\mathcal{O}_\text{V}$ as determined by the solution to $(\text{P}_{\text{V}})$ associated with the Hamiltonian system with the Okamoto Hamiltonian~\eqref{eq:H5}, i.e.\ for a generic point 
    \begin{equation}
        p_{\text{V}}^{\text{Ok}}\colon (y_2,z_2) = \left( q\,, \frac{x q'-q (x \eta +2 \kappa_{0}+\kappa_{t})+q^2 (\kappa_{0}+\kappa_{t})+\kappa_{0}}{2 q(q-1)^2 }\right) \,.
    \end{equation}
    We implement the transformation \hyperref[eq:type2_2]{$(c)$} of type II in the coordinates $(y_2,z_2) \mapsto (y,z)$:
    \begin{equation}
        \begin{cases}
            y_2 = y + q,\\[1ex]
            z_2 = z + \dfrac{x q'-q (x \eta +2 \kappa_{0}+\kappa_{t})+q^2 (\kappa_{0}+\kappa_{t})+\kappa_{0}}{2 q(q-1)^2 }.
        \end{cases}
    \end{equation}
    The system in the coordinates $  (y,z)$ is 
    \begin{equation}
        \begin{cases}
            \begin{aligned}
               y' &= z- y(2 q+y)-q'+\frac{x q'+\kappa_{0}}{2 (q-1)^2 q}\\
               &~~-\frac{x (\eta+1-2q) +\kappa_{0}(2-q)+\kappa_{t}(1-q) + q^2(2q (q-2) +x+2)}{2 (q-1)^2},
            \end{aligned}\\[2ex]
            \begin{aligned}
           z' &= \dfrac{1}{4xq^2(q-1)^3}\left(\left(x q'+\kappa_{0}\right) \left(q \left(x(3  q'-4 x)+3 \kappa_{0}\right)-\left(x q'+\kappa_{0}\right)\right)\right) \\[1ex]
            &~+\dfrac{1}{4x(q-1)^3} \left[ (\kappa_{0}+\kappa_{t})^2 -4 \kappa -3 \kappa_{0}^2+x^2(\eta ^2 +4 \eta y) -2x\left(2 \alpha  + \eta  \kappa_{t} -2 \kappa_{0} \right. \right. \\[1ex]
            &~~~ \left. \left.  +2y\left(3 \kappa_{0}  + \kappa_{t}  -2 z \right) \right)
            -2 x q \left(12 \kappa +\kappa_{0}^2 -3(\kappa_{0}+ \kappa_{t})^2+2 x q' (\kappa_{0}+\kappa_{t}+2 x) \right. \right. \\[1ex]
            &~~~~ \left.\left.+\eta  x^2 (\eta -4(y-1) +2 x (6 \alpha -6 \kappa_{0}(y-1)+\kappa_{t}(2 -\eta  -4  y)+4z(3 y -1)+3)\right) \right. \\[1ex]
            &~~~ \left. +2 x q' (2 x (y-\eta -1))-3 \kappa_{0}-\kappa_{t}+2 x (-\eta +y-1))+q^2 \left(3 \left((\kappa_{0}+\kappa_{t})^2-4 \kappa \right)-4 \eta  x^2  \right. \right. \\[1ex]
            &~~~~ \left.\left. -2 x (6 \alpha +6 \kappa_{0}+4 \kappa_{t}-2 \kappa_{0} y-2 \kappa_{t} y+12 (y-1) z+3)\right)+q^3 \left(4 \kappa -(\kappa_{0}+\kappa_{t})^2\right. \right. \\[1ex]
            &~~~~ \left.\left. +x (4 \alpha +4 \kappa_{0}+4 \kappa_{t}+8 (y-3) z+2)\right)+8 x z q^4  \right]\,,     
            \end{aligned}
            \end{cases}
    \end{equation}
and it is related to $(\text{P}_{\text{II}})$ if the coefficient function $q$ is the transcendent $\text{P}_{\text{V}}$. The latter arises as a regularising condition in both cascades originating from the base points $b_1\colon (u_0, v_0) = (0,0)$ and $b_2\colon (U_0, V_0) = (0,0)$ respectively. The regularising condition is 
\begin{equation}
\begin{split} 
    &\hspace*{-2ex}2q'''\left(x^3 q^4-2 x^3 q^3+x^3 q^2\right) +2(q')^3\left(6 x^3 q^2-4 x^3 q+ x^3\right) -3(q')^2\left(3 x^2 q^3-4 x^2 q^2+ x^2 q\right) \\[1ex]
    &\hspace*{-2ex}+q^5 \left(x^2 \eta ^2-24 \kappa - \kappa_{0}^2+6(\kappa_{0}+ \kappa_{t})^2\right)-q^3 \left(x^2 \eta ^2+4 \kappa +6 \kappa_{0}^2-(\kappa_{0} +\kappa_{t})^2\right)\\[1ex]
    &\hspace*{-2ex}+q' \left(2q^4 \left(2 x^2 \eta  (\kappa_{t}+1)- x^3 \eta ^2- x \kappa_{0}^2\right)-2q^3 \left(2 x^3 \eta ^2+2 x^2 \eta  (\kappa_{t}+1) -4 x \kappa_{0}^2\right)-2 x \kappa_{0}^2(6  q^2+4 q +1)\right)\\[1ex]
    &\hspace*{-2ex}+4q''\left(x^2 (q^4-2 q^3+ q^2)-x^3 q'\left(3  q^3-4  q^2+ q\right) \right) +q^7 \left((\kappa_{0} +\kappa_{t})^2-4 \kappa \right)\\[1ex]
    &\hspace*{-2ex}+4(q^6+q^4) \left(4 \kappa - (\kappa_{0} +\kappa_{t})^2\right)+\kappa_{0}^2 (4 q^2- q) = 0\,,
    \end{split} 
\end{equation}
which is satisfied by $q(x)$ being a solution to $(\text{P}_{\text{V}})$.

\end{example}

We finally present a last example where we manage to extend our construction to include one case of the solutions of the quasi-Painlev\'e equations in the coefficient functions.  

{
\begin{example}\label{ex:quasiP1coeff} 
We start with the system associated with $(\text{P}_{\text{IV}})$ as in~\eqref{eq:P4_system} in the variables $(y_4,z_4)$. For the orbit to consider we refer to the system for \hyperref[eq:qsiP1_syst]{$(\text{qsi-P}_{\text{I}})$} already presented in~\eqref{eq:qsiP1_syst}. We apply the transformation \hyperref[eq:type2_2]{$(c)$} of type II for the chart~$(y_4,z_4)$ along the orbit for $(\text{qsi-P}_{\text{I}})$, i.e. 
\begin{equation}
    p_{\text{qsi-I}}\colon (y_4,z_4) = \left ( q\,, q' \right)\,.  
\end{equation}
We consider the change of variables 
 \begin{equation}
        \begin{cases}
            y_4 = y  + q,\\[1ex]
            z_4 = z  + q',
        \end{cases}
    \end{equation}
generating the system in the variables $\big( y,z \big)$
\begin{equation}
    \begin{cases}
        y'=-\kappa_{0}-(q+y) \left(q+2 x+y-4 z-4 q'\right)-q',\\[2ex]
        z'=2q'( y -2 z+ q -q'+ x) +2 z( q+x+y-z) -k_1 q^4 -x -\kappa_{\infty},
    \end{cases}
\end{equation}
which is related to \hyperref[eq:P4]{$(\text{P}_{\text{IV}})$} if the coefficient function $q$ satisfies the equation \hyperref[eq:qsiP1]{$(\text{qsi-P}_{\text{I}})$}. The condition appearing at the end of one of the cascades for the system to be regular is 
\begin{equation}
    q'''-4 k_1  q^3 q' -1 = 0\,,
\end{equation}
that can be integrated in $x$, giving \hyperref[eq:qsiP1]{$(\text{qsi-P}_{\text{I}})$}.

\end{example}

}

\section{Conclusions and new perspectives} 
With the present work we significantly extend the results of Bureau and Guillot in different directions.  We show that there exist many different systems with meromorphic coefficients, this being of pivotal importance for the problems of classification of systems with the Painlev\'e property. We produced novel systems with degree greater than two related to Painlev\'e equations $(\text{P}_{\text{J}})$ different from $(\text{P}_{\text{I}})$ and $(\text{P}_{\text{II}})$, with coefficient functions being in turn solutions to $(\text{P}_{\text{J}})$ (using \hyperref[eq:tran1]{type I} transformations). It is interesting to note that the same transformation works for all Painlev\'e equations to produce Bureau-Guillot systems. 

The fact that the original system undergoing transformations of type I is related to one of the Painlev\'e equations seems to be important. For instance, as mentioned in Remark~\ref{rmk:thomas}, if we start from a system related to a quasi-Painlev\'e equation, we do not obtain the same quasi-Painlev\'e equation as a regularising condition. It would be interesting to understand how to modify the construction for this sort of extensions. Furthermore, it would be interesting to see whether this construction can be applied to systems beyond the Painlev\'e and quasi-Painlev\'e equations. This would allow us to formulate a further generalisation of the Bureau-Guillot systems for which the coefficient functions satisfy different differential equations. Possibly, the study of the expansion of solutions to the systems may provide some hints for achieving this extension. Another interesting inverse problem currently under investigation is to show that the function $h_{\text{J}}$ for the general treatment (see~\eqref{eq:general_orbit}) in the transformations of type I for the Painlev\'e equation $(\text{P}_{\text{J}})$ needs to be of a specific form, this being the only choice to produce Bureau-Guillot systems. Indeed, for instance in Proposition~\ref{Prop3.5}, by choosing the point~\eqref{eq:point_H4} for $(\text{P}_{\text{IV}})$ we do not observe any base point in the affine chart $(y,z)$. Following this line, the second order equation \eqref{eq:secord_gen} can be shown to be the fourth Painlev\'e equation (with some additional assumption also involving the form of the function $g_{\text{J}}$). 

We also notice that if we start for instance from the system for $(\text{P}_{\text{I}})$ and apply any of the transformation of type I along the orbit of $(\text{P}_{\text{I}})$, i.e.\ $(q,q')$, but assuming that $q$ is a solution for $(\text{P}_{\text{II}})$, we still see $(\text{P}_{\text{II}})$ as a regularising condition in some of the cascades. However, in this case, one of the cascades cannot be resolved in a finite number of steps. This differs from the construction proposed in Section~\ref{sec:mixedcases}, but it may be still interesting to deepen these cases.  A similar scenario takes place if we consider the transformation of type I to e.g.\ $(\text{P}_{\text{II}})$, assuming $(\text{qsi-P}_{\text{I}})$ for the coefficient function~$q(x)$. One would obtain $(\text{qsi-P}_{\text{I}})$ as a regularising condition in one cascade only, at the cost of the second cascade being infinite. This discussion is interesting from the point of view of systems of ODEs and the possibility to  construct  spaces of initial conditions after birational transformations. To obtain polynomial systems with cascades of finite lengths in the regularising procedure, we need to rely on the linear transformations of type II. Note that this transformation preserves the Hamiltonian property for systems and this allows us to generate many Hamiltonians with meromorphic coefficients.

The Bureau-Guillot systems in the sense of Definition~\ref{def:BG} may resemble an analogue of the so-called self-duality for discrete systems with Painlev\'e property~\cite{206,213}. The discrete Painlev\'e equations and their analogue of Bureau-Guillot systems are currently under investigation and the results will be reported in a separate paper.

\section*{Acknowledgements}
We would like to thank A.\ Guillot for kindly providing the most recent version of \cite{Guillot} and helpful exchanges. We also are grateful to R.\ Willox for interesting discussions, and for drawing our attention to self-dual discrete systems.

\begin{appendices}

\section{Proof of Proposition~\ref{prop:P5} for \texorpdfstring{$\text{P}_{\text{V}}$}{p5}}\label{app:P5}
{
Here we elaborate on Proposition~\ref{prop:P5}. We transform the Okamoto Hamiltonian system with Hamiltonian~\eqref{eq:H5} in $(y_5, z_5)$ via the change of variables~\eqref{eq:transf_P5}, producing the system in $(y,z)$ 
\begin{equation}
    \begin{cases}
        \begin{aligned}
        y' &= \dfrac{1}{4xq^2(q-1)^4} \left[ -2 q^5( \kappa_{0}+ \kappa_{t})^2+ q^4\left(8 \kappa_{0}^2+8 x \eta ( \kappa_{0}+\kappa_{t})+4( \kappa_{0}+ \kappa_{t})^2\right)   \right. \\[1ex]
        &~~ \left. -2q^3\left( (6 \kappa_0 + \kappa_t)^2+ 3 ( x\eta (x \eta +4 \kappa_{0}  +2 \kappa_{t}) - 8 \kappa_0^2) \right) +4q^2\left(4 \kappa_{0}(2 \kappa_{0}+\kappa_{t})+x \eta (x \eta +6\kappa_{0}  +  \kappa_{t} ) \right) \right. \\[1ex]
        &~~ \left.-2\kappa_{0}q\left(9 \kappa_{0}+4 x \eta +2 \kappa_{t} \right) +4 \kappa_{0}^2+ 2x^2(q')^2\left(2-3 q\right)-20\, q^2 y^2 z^3 \left(q^4-4 q^3+6 q^2-4 q+1 \right) \right. \\[1ex]
        &~~ \left. + 8 q^3 y^2 \left(-q^6+6 q^5-15 q^4+20 q^3-15 q^2+6 q- 1 \right)\right. \\[1ex]
        &~~ \left. + q'\left(-8 xq^3 (\kappa_{0}+ \kappa_{t}) +12 x q^2 \left( \eta  x+2 \kappa_{0} + \kappa_{t} \right) +4 x q\left(-2\eta  x-6 \kappa_{0} - \kappa_{t} \right) +8 x \kappa_{0}\right)\right. \\[1ex]
        &~~ \left.-8q y \left(
        q^6( \kappa_{0}+ \kappa_{t}) -2q^6( x \eta +3 \kappa_{0}+2 \kappa_{t}) +3q^4(2 x \eta +5 \kappa_{0}+2 \kappa_{t})-2 q^3(3 x \eta +10 \kappa_{0}+2 \kappa_{t})\right.\right. \\[1ex]
        &~~ \left. \left. +q^2(2 x \eta +15 \kappa_{0}+ \kappa_{t}) +\kappa_{0}(6  q-1)
        +x  q'\left(3 q^4-10 q^3+12 q^2-6 q+1\right) \right)\right. \\[1ex]
        &~~ \left.+z^2 \left( 16 q^2 y^2\left(-3 q^5+14 q^4-26 q^3+24 q^2-11 q+2 \right)-16y\left(q^5( \kappa_{0}+ \kappa_{t}) -q^4( x \eta +4 \kappa_{0}+3 \kappa_{t})\right. \right. \right.\\[1ex]
        &~~ \left. \left. \left. +q^3(2 x \eta +6 \kappa_{0}+3 \kappa_{t}) -q^2(x \eta +4 \kappa_{0}+ \kappa_{t}) + \kappa_{0} q+xq'\left( q^3-2 q^2+ q\right) \right) \right) \right. \\[1ex]
        
        &~~  \left.  -3z \left( 
        q^4 (\kappa_{0}+ \kappa_{t})^2 -2q^3\left(2 \kappa_{0}^2+3 \kappa_{t} \kappa_{0}+2 \kappa_{t}^2 + x \eta  (  \kappa_{0}+\kappa_{t})\right) -2\kappa_{0} q\left(2 \kappa_{0}+ x \eta  + \kappa_{t} \right)+ (\kappa_{0}^2+ x^2 (q')^2)\right. \right.\\[1ex]
        &~~ \left. \left.  + q^2\left(x^2 \eta ^2+2 x\eta(2 \kappa_{0}  + \kappa_{t} ) + \kappa_{t}^2+6 \kappa_{0} (\kappa_{0}+\kappa_{t})\right)  
        +4 q^2 y^2 \left( 3q^6-16 q^5+35 q^4-40 q^3+25 q^2-8 q+1 \right)  \right. \right.\\[1ex]
        &~~  \left. -2 q'\left(x \left(-q^2(  \kappa_{0}+ \kappa_{t})+q\left( \eta  x+2 \kappa_{0} + \kappa_{t} \right) - \kappa_{0}\right) +4y \left(q^6( \kappa_{0}+\kappa_{t}) - q^5(3 x \eta +10 \kappa_{0}+7 \kappa_{t}) \right. \right. \right.\\[1ex]
        &~~ \left. \left. \left. + q^4(8 x \eta +20 \kappa_{0}+9 \kappa_{t}) -q^3(7 x \eta +20 \kappa_{0}+5 \kappa_{t}) + q^2(2 x \eta +10 \kappa_{0}+ \kappa_{t})-2 \kappa_{0} q\right. \right. \right.\\[1ex]
        &~~ \left. \left. \left.+ x q q'\left(3 q^3-8 q^2+7 q-2 \right) \right)\right)   \right],
        \end{aligned} \\
        \\[-1ex]
        \begin{aligned} 
        z' & = \dfrac{z}{x q(q-1)^2}\left[ 
        \kappa_{0}+z^2 \left( \kappa_{0}+x q'-q (2 \kappa_{0}+\kappa_{t}+\eta x )+q^2(\kappa_{0}+\kappa_{t}) +2 q y \left(3 q^3-8 q^2+7 q-2 \right)\right)\right.\\[1ex]
        &~~ \left.+z \left(-2 \kappa_{0}+(3 x q-2 x) q'-3q^2 (2 \kappa_{0}+ \kappa_{t}+ \eta  x)+q (6 \kappa_{0}+\kappa_{t}+2 \eta  x)+2q^3(\kappa_{0}+ \kappa_{t}) \right. \right.\\[1ex]
        &~~ \left. \left. 
        +2qy \left(3 q^4-10 q^3+12 q^2-6 q+ 1\right)\right)+xq'\left(3 q^2-4 q+1\right) -2q^3 (2 \kappa_{0}+ \kappa_{t}+ \eta  x)\right. \\[1ex]
        &~~  \left. +q^2 (6 \kappa_{0}+\kappa_{t}+2 \eta  x)+q^4(\kappa_{0}+\kappa_{t}) -4 \kappa_{0} q+y z^3 \left(2 q^3-4 q^2+2 q\right)\right. \\[1ex]
        &~~ \left. +y \left(2 q^6-8 q^5+12 q^4-8 q^3+2 q^2\right)
        \right].
        \end{aligned}
    \end{cases}
\end{equation}
Analysing the system in $(y,z)$ in the compactified space $\mathbb{P}^2$, we identify two base points $b_1\colon (u_0\,,v_0)=(0\,,0)$ and $b_2\colon (U_0\,,V_0)=(0\,,0)$. Both base points give rise to cascades splitting into two branches. The final charts of the cascade originated in $b_1$ appear to be already regularised, while at the end of both final charts from the cascade originated from $b_2$ we find the same condition for the system to be regular: 
\begin{equation}
    \begin{aligned}
    &q^4 \left(-2 \eta  (\kappa_{t}+1)+q' \left(\kappa_{0}^2+\eta  x (-2 \kappa_{t}+\eta  x-2)-1\right)-x \left(x q'''+3 q''\right)\right)\\[1ex]
        &+q^3 \left(6 x (q')^2+2 x \left(x q'''+3 q''\right)+2 q' \left(-2 \kappa_{0}^2+3 x^2 q''+\eta  x (\kappa_{t}+\eta  x+1)+1\right)+\eta  (\kappa_{t}+\eta  x+1)\right)\\[1ex]
        &-q^2 \left(6 x^2 (q')^3+8 x (q')^2+x \left(x q'''+3 q''\right)+q' \left(-6 \kappa_{0}^2+8 x^2 q''+1\right)\right)\\[1ex]
        &+\eta  q^5 (\kappa_{t}+\eta  (-x)+1) -x^2 (q')^3+\kappa_{0}^2 q'+2 q q' \left(x \left(x q''+2 x (q')^2+q'\right)-2 \kappa_{0}^2\right) =0\,.
    \end{aligned}
\end{equation}
The condition is identically satisfied if $q(x)$ is a $\text{P}_{\text{V}}$ transcendent.

}

\section{Proof of Proposition~\ref{prop:P6} for \texorpdfstring{$\text{P}_{\text{VI}}$}{p6}}\label{app:P6}
{
We report here the system obtained by applying the  transformation~\hyperref[eq:tran1]{$(b)$} of type I~\eqref{eq:transf_P6} to the chart~$\big( y_6, z_6 \big)$, obtaining the system in the variables $(y,z)$: 
\begin{equation}
    \begin{cases}
        \begin{aligned}
            y'&=\frac{1}{4 x (x-q)^2 (q-1)^2 q^2}\left[16 q y z^2 (x-q) (q-1)  \left(3  q^4y-4 q^3 y  (x+1) -20 y^2 (x-q)^2 (q-1)^2 q^2 z^3 \right. \right.\\[1ex]
            &~~ \left. \left.  +(q^2 x y (x+5) +y+\kappa_{0}+\kappa_{1}+\kappa_{t}-1) -q y(x (x+1) +\kappa_{0}(x+1)+x \kappa_{1}+\kappa_{t}-1) +x \left(\kappa_{0}+q'\right)\right) \right. \\[1ex]
            &~~  \left.  +3 \left(-12 y^2 q^8+32 q^7 y^2  (x+1) -4 y (7 y(x (x+3)+1) +2q^6(\kappa_{0}+\kappa_{1}+\kappa_{t}-1)) \right. \right. \\[1ex]
            &~~ \left. \left. +4 q^5 y (2 y (x+1) (x (x+8)+1) +5 (\kappa_{0}+ \kappa_{t})+x (5 (\kappa_{0}+ \kappa_{1}-1)+2 (\kappa_{1}+ \kappa_{t}-1)) \right. \right. \\[1ex]
            &~~ \left. \left.  -q^4\left((\kappa_{0}+\kappa_{1}+\kappa_{t}-1)^2+4 y \left(5 x y (x (x+3)+1) +4 (\kappa_{0}+\kappa_{t}-1)+x (4 (\kappa_{0}(x+3) + \kappa_{1}x) \right. \right. \right. \right. \\[1ex]
            &~~ \left. \left. \left. \left. +5 (\kappa_{1}+\kappa_{t}-1))+3 x q'\right)\right) +2 q^3\left((\kappa_{0}+\kappa_{1}+\kappa_{t}-1) (\kappa_{0}(x+1)+x \kappa_{1}+\kappa_{t}-1) \right. \right. \right.  \\[1ex]
            &~~ \left. \left. \left.  +2 y \left(\kappa_{0}+x ((x (x+9)+9) \kappa_{0}+x (x+4) \kappa_{1}+4 (\kappa_{t}-1))+\kappa_{t}+4 x (x+1) \left(x y+q'\right)-1\right)\right) \right. \right.   \\[1ex]
            &~~ \left. \left. -q^2\left((x (x+4)+1) \kappa_{0}^2+2 \left(\kappa_{1} x^2+2 (\kappa_{1}+\kappa_{t}-1) x+\kappa_{t}-1\right) \kappa_{0}+(x \kappa_{1}+\kappa_{t}-1)^2\right. \right. \right.  \\[1ex]
            &~~ \left. \left. \left.+4 x y z \Big(y x^2+\kappa_{1} x^2+2 (x (x+3)+1) \kappa_{0}+\kappa_{t}-1\right)+2 x (2 (x (x+5)+1) y+\kappa_{0}+\kappa_{1}+\kappa_{t}-1) q'\right) \right.   \\[1ex]
            &~~   \left.+2 q x (2 x (x+1) y+x \kappa_{0}+\kappa_{0}+x \kappa_{1}+\kappa_{t}-1) \left(\kappa_{0}+q'\right) -x^2 \left(\kappa_{0}+q'\right)^2\Big) \right.   \\[1ex]
            &~~   \left. +2 \left(-4  q^9 y^2+12 q^8 y^2 (x+1)  -4 q^7 y (3 (x (x+3)+1) y+\kappa_{0}+\kappa_{1}+\kappa_{t}-1) \right.\right.   \\[1ex]
            &~~   \left.\left.  +4 y (y(x+1) (x (x+8)+1) +3 (\kappa_{0}+ \kappa_{t})+\kappa_{1}+q^6 x (3 (\kappa_{0}+ \kappa_{1}-1)+\kappa_{t}-1)) \right.\right.   \\[1ex]
            &~~   \left.\left. -q^5\left((\kappa_{0}+\kappa_{1}+\kappa_{t}-1)^2+12 y \left(\kappa_{0}+\kappa_{t}+x ((x+3) \kappa_{0}+\kappa_{1}(x+1)+\kappa_{t}-1)\right.\right.\right. \right.   \\[1ex]
            &~~   \left.\left.\left.\left.+x \left(x (x+3) y+y+q'\right)-1\right)\right) +q^4\left((\kappa_{0}+\kappa_{1}+\kappa_{t}-1) (3 \kappa_{0}-\kappa_{1}+x (3 (\kappa_{0}+3 \kappa_{1}+\kappa_{t}-1)-\kappa_{t}+1)) \right.\right.\right.    \\[1ex]
            &~~   \left.\left.\left. +4 y \left(\kappa_{0}+x ((\kappa_{0}x (x+9)+9) +\kappa_{1}x (x+3) +3 (\kappa_{t}-1))+\kappa_{t}+x (x+1) \left(3 x y+5 q'\right)-1\right)\right) \right.\right.   \\[1ex]
            &~~   \left.\left. -q^3\left(3 x^2 (\kappa_{0}+\kappa_{1})^2+3 (\kappa_{0}+\kappa_{t}-1)^2+2 x \left(4 \kappa_{0}^2+3\kappa_{0} (\kappa_{1}+\kappa_{t}-1)  -(\kappa_{1}+1)(\kappa_{t}+\kappa_{1} - \kappa_{t})+\kappa_{t}-1\right)\right.\right.\right.   \\[1ex]
            &~~   \left.\left.\left.  +4 x \left(y \left(y x^2+\kappa_{1} x^2+3 (x (x+3)+1) \kappa_{0}+\kappa_{t}-1\right)+(2 (x (x+4)+1) y+\kappa_{0}+\kappa_{1}+\kappa_{t}-1) q'\right)\right)\right.\right.   \\[1ex]
            &~~   \left.\left. +q^2\left(12x^2 y (x+1)  \left(\kappa_{0}+q'\right) +(x \kappa_{0}+\kappa_{0}+x \kappa_{1}+\kappa_{t}-1) \left(\kappa_{0}+x (\kappa_{0}(x+6) + \kappa_{1}(x-1)-\kappa_{t}+1) \right.\right. \right.  \right.\\[1ex]
            &~~   \left.\left. \left. \left. +\kappa_{t}+6 x q'-1\right)\right) -x q \left(\kappa_{0}+q'\right) \left((x+2) (2 x+1) \kappa_{0}+2 \left(\kappa_{1} x^2+\kappa_{t}-1\right)+x \left(4 x y+3 q'\right)\right)\right.\right. \\[1ex]
            &~~   \left. \left. +x^2 (x+1) \left(\kappa_{0}+q'\right)^2\right)\right],        \end{aligned} \\
        \\ 
        \begin{aligned}
        z'&=\frac{1}{x (x-q) (q-1) q}\left[z \left(z \left(x (-3 q+x+1) q'-q \left(\kappa_{0}+\kappa_{t}+\kappa_{1} x^2+\kappa_{0} (x+4) x-1\right)\right.\right. \right.\\[1ex]
            &~~   \left. \left. \left. -2 q^3 (\kappa_{0}+\kappa_{1}+\kappa_{t}-1)+3 q^2 (\kappa_{0}+\kappa_{t}+\kappa_{0} x+\kappa_{1} x-1)-2 y q (q-1) (q-x) \left(3 q^2-2 (x+1) q+x\right) \right.\right. \right.\\[1ex]
            &~~   \left. \left. \left. +\kappa_{0} x (x+1)\right)-z^2 \left(x q'+q ( (q-1)(\kappa_{0}+\kappa_{t})+\kappa_{1} q+2 y (x-q) (q-1) (-3 q+x+1) \right.\right. \right.\\[1ex]
            &~~   \left. \left. \left. - x(\kappa_{0}+\kappa_{1})-q+1)+\kappa_{0} x\right)+x ((-3 q+2 x+2) q-x) q'+q \left(-q \left(\kappa_{0}(x^2+4x+1)+\kappa_{t}+\kappa_{1} x^2-1\right) \right.\right. \right.\\[1ex]
            &~~   \left. \left. \left.-\left(q^3 (\kappa_{0}+\kappa_{1}+\kappa_{t}-1)\right)+2 q^2 (\kappa_{0}(x+1)+\kappa_{t}+\kappa_{1} x-1)-2 y (x-q)^2 (q-1)^2 q+2 \kappa_{0} x (x+1)\right) \right.\right. \\[1ex]
            &~~    \left. \left. +2 y z^3 (x-q) (q-1) q-\kappa_{0} x^2\right)\right].
        \end{aligned}
    \end{cases}
\end{equation}
We identify the base points in $\mathbb{P}^2$, these being the origins of the two affine charts $b_1\colon (u_0,v_0)=(0,0)$ and $b_2\colon (U_0,V_0)=(0,0)$. The system here is associated with the equation $(\text{P}_{\text{VI}})$ with the regularising condition:
\begin{equation}
\begin{aligned}
    &2 x^4 q' \left(q'^2-\kappa_{0}^2\right)+x^3 q \left(\left(6 (x+1) q'+1\right) \left(\kappa_{0}^2-q'^2\right)-4 x q' q''\right)\\[1ex]
    &-q^5 \left(12 x q'^2+4 x \left(x (x+1) q^{(3)}+(2 x+3) q''\right)+2 q' \left(6 x^2 q''+x \left(-3 \kappa_{0}^2-3 \kappa_{1}^2+3 \kappa_{t}(\kappa_{t}-4) \right. \right. \right. \\[1ex]
    & \left.\left.\left. -3 x \left(\kappa_{0}^2-\kappa_{1}^2+\kappa_{t}^2-2 \kappa_{t}+1\right)+8\right)+2\right)-3 \left(\kappa_{0}^2+4 \kappa_{t}(1-2x)+x \left(\kappa_{0}^2 (x+3)-\kappa_{1}^2 (x+1)+5\right)  \right. \right.  \\[1ex]
    & \left.\left. +\kappa_{t}^2 (3 x-1)-3\right)\right)+q^4 \left(-\kappa_{0}^2-4 \kappa_{t}+12 x^2 q'^3+4 x (2 x+5) q'^2+2 x \left(x (x (x+4)+1) q^{(3)}\right. \right.  \\[1ex]
    & \left.\left.+(x (x+8)+3) q''\right)+2 q' \left(10 (x+1) x^2 q''+x \left(-3 \kappa_{0}^2+3 (\kappa_{t}-4) \kappa_{t}+3 x^2 \left(\kappa_{1}^2-\kappa_{0}^2\right)\right. \right.\right.  \\[1ex]
    & \left.\left.\left.-3 x \left(3 \kappa_{0}^2+\kappa_{1}^2+(\kappa_{t}-1)^2\right)+7\right)+1\right)+x \left(\kappa_{1}^2 x (x+3)-\kappa_{0}^2 (x (x+9)+9)-11\right)\right.   \\[1ex]
    & \left.+\kappa_{t}^2 (1-5 x)+16 \kappa_{t} x+3\right)+x q^3 \left(\kappa_{t}^2-4 \kappa_{t}-4 x^2 (x+1) q^{(3)}-4 x (x+2) q''-16 x (x+1) q'^3\right.   \\[1ex]
    & \left.-(x (2 x+11)+8) q'^2+2 q' \left(\kappa_{0}^2+4 \kappa_{t}-4 x (x (x+4)+1) q''+x \left(\kappa_{0}^2 (x (x+9)+9)-\kappa_{1}^2 (x-1) x+1\right)\right. \right.  \\[1ex]
    & \left.\left.+\kappa_{t}^2 (x-1)-2 \kappa_{t} x-2\right)-\kappa_{1}^2 x^2+3 \kappa_{0}^2 (x (x+3)+1)+3\right)+x^2 q^2 \left(6 (x (x+3)+1) q'^3+3 (x+1) q'^2\right.   \\[1ex]
    & \left.+2 x \left(x q^{(3)}+q''\right)+6 q' \left(2 x (x+1) q''-\kappa_{0}^2 (x (x+3)+1)\right)-3 \kappa_{0}^2 (x+1)\right)+q^6 \left(-3 \kappa_{0}^2+\kappa_{1}^2-12 \kappa_{t}\right. \\[1ex]
    & \left.-2 q' \left(\kappa_{1}^2+x \left(\kappa_{0}^2-\kappa_{1}^2+\kappa_{t}^2 (x-1)-2 \kappa_{t} (x-2)+x-3\right)-1\right)+2 x \left(x q^{(3)}+3 q''\right)\right.   \\[1ex]
    & \left.-3 x \left(\kappa_{0}^2-\kappa_{1}^2+3\right)+\kappa_{t}^2 (3-7 x)+16 \kappa_{t} x+9\right)+q^7 \left(\kappa_{0}^2-\kappa_{1}^2-\kappa_{t}^2+4 \kappa_{t}+2 (\kappa_{t}-1)^2 x-3\right)=0\,,
    \end{aligned} 
\end{equation}
which is satisfied if $q(x)$ is a $\text{P}_{\text{VI}}$ transcendent.

}

\end{appendices}

\section*{Conflict of interest}
The authors have no conflict of interest to declare that is relevant to the content of this article.

\section*{Data availability}
Data sharing is not applicable to this article as no new data were created or analysed in this study.

\bibliography{thebiblio}
\bibliographystyle{style_alpha}

\end{document}